\documentclass[letterpaper, 10 pt, conference]{ieeeconf}  

\IEEEoverridecommandlockouts                              

\usepackage{amsthm}

\usepackage{amsmath} 
\usepackage{amssymb}  

\usepackage{dsfont}
\usepackage{multirow}
\usepackage{algorithm}
\usepackage[noend]{algpseudocode}
\usepackage{subcaption}
\usepackage{xspace}
\usepackage{cite}
\usepackage{graphicx}

\DeclareFontFamily{U}{mathx}{\hyphenchar\font45}
\DeclareFontShape{U}{mathx}{m}{n}{
      <5> <6> <7> <8> <9> <10>
      <10.95> <12> <14.4> <17.28> <20.74> <24.88>
      mathx10
      }{}
\DeclareSymbolFont{mathx}{U}{mathx}{m}{n}
\DeclareMathSymbol{\bigtimes}{1}{mathx}{"91}

\newtheorem{assumption}{Assumption}
\newtheorem{remark}{Remark}
\newtheorem{definition}{Definition}
\newtheorem{proposition}{Proposition}
\newtheorem{lemma}{Lemma}
\newtheorem{problem}{Problem}
\newtheorem{theorem}{Theorem}
\newtheorem{corollary}{Corollary}

\newcommand{\reals}{\mathbb{R}}
\newcommand{\naturals}{\mathbb{N}}

\newcommand{\Lcal}{\mathcal{L}}
\newcommand{\Scal}{\mathcal{S}}

\newcommand{\Pcal}{\mathcal{P}}

\newcommand{\Tau}{\mathcal{T}}
\newcommand{\Vcal}{\mathcal{V}}
\newcommand{\Ucal}{\mathcal{U}}
\newcommand{\Acal}{\mathcal{A}}

\newcommand{\strategy}{\pi_s}

\newcommand{\A}{\mathcal{A}}
\newcommand{\LTLf}{LTL$_f$\xspace}

\newcommand{\AP}{\mathrm{AP}}
\newcommand{\Prop}{\mathfrak{p}}
\newcommand{\safe}{\mathrm{safe}}
\newcommand{\unsafe}{\text{unsafe}}

\newcommand{\Xsafe}{X_\safe}
\newcommand{\Xunsafe}{X_\unsafe}

\newcommand{\abs}{\text{abs}}
\newcommand{\Xabs}{X_\abs}

\newcommand{\Bcal}{\mathcal{B}}
\newcommand{\Fcal}{\mathcal{F}}
\newcommand{\W}{\mathcal{W}}

\usepackage{xcolor}

\newcommand{\ig}[1]{\textcolor{blue}{[IG: #1]}}
\newcommand{\ml}[1]{\textcolor{cyan}{[ML: #1]}}

\title{\LARGE \bf
On the Optimality of Uncertain MDP Abstractions
}

\author{Ib\'on Gracia and Morteza Lahijanian
\thanks{Authors are with the Department of Aerospace Engineering Sciences,
        University of Colorado Boulder, 3775 Discovery Dr., Boulder, Colorado
        {\tt\small \{ibon.gracia, morteza.lahijanian\}@colorado.edu}}%
}

\begin{document}

\maketitle
\thispagestyle{empty}
\pagestyle{empty}


\begin{abstract}

We study the asymptotic optimality of abstraction-based control synthesis algorithms. Specifically, we consider uncertain MDP (UMDP) abstraction, and investigate whether refinement leads to optimal results, i.e., an optimal controller and zero error bound. Additionally, we study completeness of abstraction-refinement algorithms, i.e., that the algorithm produces near-optimal results in finite time. The focus is on nonlinear stochastic systems with general vector fields and temporal logic specifications. We present an algorithm that abstracts the system into a UMDP and synthesizes a controller with performance guarantees via robust dynamic programming. Then, the algorithm iteratively refines the abstraction until a near-optimality criterion is met. A thorough theoretical analysis reveals a sufficient condition, which we denote \emph{vanishing ambiguity}, guaranteeing asymptotic optimality of the abstraction process and completeness of the algorithm. We show that set-valued MDP abstractions satisfy this criterion, whereas interval MDP abstractions lack such a guarantee.

%
\end{abstract}


\section{Introduction}

Automatic control of \emph{safety-critical} cyber-physical systems requires formal guarantees, but in practice these rely on approximations with quantified (and ideally tight) error bounds. This is particularly challenging for complex, nonlinear, and stochastic robotic systems, where loose bounds (e.g., an autonomous car with safety probability lower-bound $0.90$ and upper-bound $1.0$) are insufficient for deployment.
Abstraction-based methods enable formal synthesis for such systems, and even for complex \emph{temporal logic} specifications, by modeling the systems as finite-state \emph{uncertain MDPs} (UMDPs). While these approaches provide sound probability bounds, they typically lack guarantees of asymptotic optimality, i.e., that refinement of the abstraction converges to the true satisfaction probability. Indeed, finer discretizations do not always improve bounds and can even degrade them \cite{Gracia:L4DC:2025}.
This raises a fundamental question: ``\emph{under what conditions are UMDP abstraction-based methods asymptotically optimal?}'' In this work, we aim to address it by identifying a sufficient condition under which asymptotic optimality is guaranteed and UMDP classes that satisfy this condition.

When it comes to formal control via UMDP abstractions, work \cite{abate2010approximate} focuses on stochastic hybrid systems with regular transition and reset kernels. It abstracts the continuous-state system as an MDP and then quantifies the error between the safety probabilities of the original system and the abstraction. The obtained error bound vanishes as the discretization becomes finer, meaning that the abstraction approach is asymptotically optimal. However, the obtained bounds are loose. In consequence, \cite{lahijanian2015formal} proposes to soundly abstract the original system into a UMDP class named \emph{interval MDP} (IMDP). This approach effectively embeds the discretization error into the ambiguous, i.e., not fixed, transitions of the abstraction. 
Such works obtain the tightest bounds on the transition probabilities under restrictive assumptions on the system's dynamics, like additive Gaussian disturbances and/or linear dynamics \cite{cauchi2019efficiency,adams2022formal}. 
In consequence, these abstraction processes are also asymptotically optimal. Additional works also obtain IMDP abstractions assuming additive and/or unimodal disturbances \cite{dutreix2022abstraction,jiang2022safe, van2022temporal}. On the other hand, for general nonlinear dynamics, most works employ $1$-step reachable sets to obtain sound abstractions \cite{gracia2024data,nazeri2025data,10347405, gracia2025beyond}. While these approaches yield correct bounds in the satisfaction probabilities, it is not clear if they are asymptotically optimal.

Recent works \cite{gracia2024data,gracia2025efficient,gracia2025beyond,Gracia:L4DC:2025, coppola2024enhancing} employ more general UMDP abstraction classes than IMDPs to account for additional, e.g., epistemic, uncertainties \cite{gracia2024data,gracia2025efficient}, or to obtain tighter performance bounds \cite{gracia2025beyond,Gracia:L4DC:2025, coppola2024enhancing}. 
Some works also propose heuristic abstraction-refinement algorithms \cite{esmaeil2013adaptive,dutreix2022abstraction,adams2022formal} in order to improve the obtained guarantees. However, a discussion on asymptotic optimality is again missing in these works. In fact, to the best of our knowledge, no UMDP abstraction-based approach ensures asymptotic optimality when considering temporal logic specifications, as the value functions become discontinuous, making convergence analysis non-trivial \cite{abate2010approximate}.
%

In this work we lay the foundation for the general theory of asymptotic optimality in UMDP abstractions under specifications given as \emph{linear temporal logic over finite traces} (\LTLf) \cite{Giacomo2013} formulas. We considering a fully known system model and general rectangular UMDPs \cite{iyengar2005robust}. Our contributions are threefold: (i) we derive a sufficient condition on the abstraction procedure that guarantees its asymptotic optimality; (ii) we show that, for general nonlinear dynamics and when the abstraction is obtained via reach set computations, state-of-the-art \emph{set-valued MDP} (SMDP) abstractions~\cite{gracia2025beyond}, are indeed asymptotically optimal, while others, relying on the popular IMDP abstractions~\cite{10347405}, are not; (iii) we propose a general abstraction-refinement algorithm that synthesizes $\varepsilon$-optimal controllers and satisfaction probability bounds in finite time for any asymptotically optimal UMDP abstraction.

\section{Basic Notation}

We let $\naturals_0 := \naturals\cup\{0\}$. Given a set $A\subseteq X$, we let $\partial A$ denote the boundary of $A$ and denote by $\mathds{1}_A:X\mapsto \{0,1\}$ the indicator function of $A$, which returns $\mathds{1}_A(x) = 1$ if $x\in A$ and $0$ otherwise. Given a subset $A$ of a metric space $(X,d)$, we let $\mathrm{diam}_d(A) := \sup_{x,x'\in A} d(x,x')$ be the diameter of $A$. We also denote by $B(0,\epsilon) \subset X$ the closed $\epsilon$-ball centered at $0$. Given subsets $A, B \subset X$, we denote by $A \oplus B$ their \emph{Minkowski sum}. 
%
Given a set $X$, we denote by $\Fcal(X)$ its \emph{Borel $\sigma$-algebra} and by $\Pcal(X)$ the set of \emph{Borel probability distributions} (measures) over $X$. Given a measurable space $(X, \Fcal(X))$, we denote by $\mathcal{L}$ the \emph{Lebesgue measure}. When the space is equipped with a metric $d$, we denote the \emph{Wasserstein distance} between measures $\gamma,\gamma'\in\Pcal(X)$ by $\W(\gamma, \gamma') := \inf_{\pi\in \Pi(\gamma,\gamma')} \int_{X^2}d(x,x')d\pi(x,x')$, where $\Pi(\gamma,\gamma') \in \Pcal(X^2)$ is the set of probability measures on $X^2$ with marginals $\gamma$ and $\gamma'$. 
Given a subset $A$ of $(X, d)$, the measurable space $(A,\Fcal(A))$ and the probability measure $P \in \Pcal(A)$, with a slight abuse of notation, we also interpret this measure as an element of $\Pcal(X)$, without formally defining its extension on $(X, \Fcal(X))$. In consequence, we sometimes write $\int_X f dP := \int_A f|_A dP$ for $f : X\rightarrow \reals$, with $f|_A$ being the restriction of $f$ to $A$. 

\section{Problem Formulation}

We consider a discrete-time stochastic system of the form
\begin{align}
\label{eq:sys}
    x_{t+1} = f(x_t,u_t,w_t),
\end{align}
with $x_t \in X\subseteq \reals^n$ being the state at time $t\in\naturals_0$, $u_t \in U$ denoting the action where $U$ is a finite set, and $w_t$ is a random disturbance. We assume that $\{w_t\}_{t\in\naturals_0}$ is an i.i.d. stochastic process taking values in a compact set $W \subset \reals^d$, where $w_t \sim P_W$. Given $x_0,\ldots, x_T \in X$, $u_0,\dots,u_{T-1}\in U$, and $T \in \naturals_0$, we denote a finite \emph{trajectory} of System~\eqref{eq:sys} by $\omega^X =x_0 \xrightarrow{u_0}
\ldots \xrightarrow{u_{T-1}} x_T$. We let $\Omega^X_T$ denote the set of all trajectories of length $T \in \naturals_0$. A time-varying \emph{controller} $\kappa := (\kappa_0,\dots, \kappa_{T-1})$ is a set of functions indexed by time where each $\kappa_t$ is a function $\kappa_t : \omega_t^X \rightarrow U$ that maps the current trajectory to the next action $u_t\in U$. We denote by $\mathcal{K}$ the set of all such controllers. 
%
The \emph{transition kernel} of System~\eqref{eq:sys} is a function $\Tau:\Fcal(X)\times X\times U\rightarrow [0,1]$ that, given current state $x_t$ and action $u_t$, and the Borel set $D$, returns the probability that the successor state belongs to $D$, i.e., $x_{t+1} \in D $ with probability $\Tau(D \mid x,u)  := \int_W \mathds{1}_D(f(x,u,w))dP_W(w)$.

We are interested in System~\eqref{eq:sys} remaining within the safe set $\Xsafe$ while satisfying certain temporal properties regarding a finite set of regions of interest $R := \{r_1,\dots, r_{|R|}\}$, where $r_1,\dots,r_{|R|-1} \subset X$ and $r_{|R|} = X\setminus \Xsafe =: X_\unsafe$. Without loss of generality, we assume that $R$ is a partition of $X$. We also consider a set of \emph{atomic propositions} $\rm{AP} := \{\Prop_1,\dots,\Prop_{|\rm{AP}|}\}$, and let $L:X\rightarrow 2^{\rm{AP}}$ be the labeling function, mapping a state $x$ to the set of atomic propositions that are ``true'' at region $r \ni x$. For example, we let $\Prop_1 := \Prop_\safe \in L(x)$ for all $x \in \Xsafe$. Consequently, each trajectory $\omega^X_T = x_0 \xrightarrow{u_0} \ldots \xrightarrow{u_{T-1}} x_T$
results in the (observation) \emph{trace} 
$l = l_0\dots l_T, $
where $l_t := L(x_t)$.
Formally, we consider specifications given as \LTLf formulas \cite{Giacomo2013} over the set $\AP$, which are able to unambiguously characterize temporal behaviors of the system. 

An \LTLf formula is recursively defined as
\begin{align*}
    \varphi := \top \mid \Prop \mid \neg \varphi \mid \varphi_1 \land \varphi_2 \mid \bigcirc\varphi_1 \mid \varphi_1 \mathbin{\mathrm{U}} \varphi_2,
\end{align*}
where $\Prop \in \AP$, $\neg$ (negation) and $\land$ (and) are boolean operators, $\bigcirc$ (next) and $\mathrm{U}$ (until) are temporal operators, and $\varphi_1,\varphi_2$ are \LTLf formulas. From $\mathrm{U}$, we define 
$\Diamond$ (eventually) and
$\square$ (globally) as $\Diamond \varphi \equiv \top \mathrm{U} \varphi$ and $\square \varphi \equiv \neg \Diamond \neg \varphi$. \LTLf formulae are semantically interpreted over finite traces \cite{de2013linear}: we say that trajectory $\omega^X$ satisfies formula $\varphi$, i.e., $\omega^X \models \varphi$, if some prefix of its trace $l$ satisfies $\varphi$. Specifically, we 
consider formulas of the form $\square \Prop_{\rm{safe}}\land \varphi$, where $\varphi$ is an \LTLf formula. 
Additionally, in this work, we consider the problem of satisfying an \LTLf formula in finite time. 
%
\begin{definition}[System]
    We refer by \emph{the system} to the tuple $\mathcal{S} := (X,U,W,f,P_W,R,\Xsafe,\AP,L)$. Given a controller $\kappa$, we denote by $\mathrm{Pr}_{x_0}^\kappa$ the probability measure~\cite{bertsekas1996stochastic} over the finite trajectories of $\mathcal{S}$ under $\kappa$ with initial state $x_0\in X$.
\end{definition}
Our results require the following three standard assumptions on $\Scal$:
\begin{assumption}
\label{ass:measure_0}
    The set $X_\safe$ is compact. Additionally, $\Lcal(\partial r) = 0$ for all $r\in R$.
\end{assumption}

\begin{assumption}
    \label{ass:density}
    For all $x\in X$ and $u\in U$, the conditional transition probability distribution $\Tau(\cdot\mid x, u)$ is continuous with respect to the Lebesgue measure, and has a conditional density $\tau(x' \mid x,u)$ which is bounded: $\overline \tau := \sup_{x,x'\in X, u\in U}\tau(x'\mid x, u) < \infty$.
\end{assumption}
\begin{assumption}
\label{ass:lipschitz}
    The vector field $f$ is globally Lipschitz continuous in its first and third arguments, i.e., there exists $L_f \ge 0$ such that $\|f(x,u,w) - f(x', u, w')\| \le L_f(\|x - x'\| + \|w - w'\|)$ for all $x, x' \in X$, $w, w'\in W$, $u\in U$.
\end{assumption}
%

Given a controller $\kappa$, initial state $x_0 \in X$, an \LTLf formula $\varphi$, and a time horizon $T \in \naturals$, we denote by $\mathrm{Pr}_{x_0}^\kappa[\varphi, T]$ the probability that the trajectory of system $\mathcal S$ from $x_0$ under $\kappa$ satisfies $\varphi$ within $T$ time-steps. Similarly, we denote by $\mathrm{Pr}_{x_0}^*[\varphi, T] := \sup_{\kappa \in \mathcal{K}} \mathrm{Pr}_{x_0}^\kappa[\varphi, T]$ the optimal satisfaction probability. 

In this work we consider the problem of synthesizing a controller that ensures a near-optimal satisfaction probability, and obtaining tight bounds in the satisfaction probability.
\begin{problem}
\label{prob:prob}
    Given system $\mathcal S$, an \LTLf specification $\varphi = \square \Prop_\safe \land \varphi'$, a time horizon $T \in\naturals_0$ and an optimality margin $\varepsilon > 0$, synthesize a controller $\kappa^\varepsilon$ and obtain performance bounds $\underline p^\varepsilon$ and $\overline p^\varepsilon$ that ensure, for all $x_0 \in X$:
    \begin{itemize}
        \item[(i)] $|\mathrm{Pr}_{x_0}^{\kappa^\varepsilon}[\varphi, T] - \mathrm{Pr}_{x_0}^*[\varphi, T]|\le \varepsilon$,
        \item[(ii)] $\mathrm{Pr}_{x_0}^{\kappa^\varepsilon}[\varphi, T] \in [\underline p^\varepsilon(x_0), \overline p^\varepsilon(x_0)]$,
        \item[(iii)] $\overline p^\varepsilon(x_0) - \underline p^\varepsilon(x_0) \le \varepsilon$.
    \end{itemize}
\end{problem}

\paragraph*{Overview of the Approach}

We approach Problem~\ref{prob:prob} through an abstraction-refinement framework and mainly focus on analyzing the conditions under which termination is guaranteed.  
In this framework, at each iteration, we partition the state space and construct a UMDP abstraction of $\Scal$ that captures all relevant behaviors. We then synthesize a controller via robust dynamic programming, which also provides bounds on the satisfaction probability; thus satisfying condition (ii), as in existing UMDP-based methods \cite{10347405,gracia2025beyond}. If the bounds meet the termination criterion, conditions (i)-(iii) are satisfied; otherwise, we refine the partition and repeat.

However, refinement does not always tighten the bounds. To address this, we introduce a sufficient condition, \emph{vanishing ambiguity}, under which the algorithm is \emph{complete}, i.e., it terminates in finite time for any $\varphi'$ and $\varepsilon$. This condition further ensures \emph{asymptotic optimality}: as the partition becomes infinitely fine, all conditions in Problem~\ref{prob:prob} are satisfied with $\varepsilon = 0$. Finally, we analyze completeness for different abstractions, showing that SMDPs guarantee asymptotic optimality, whereas IMDPs do not. All the proofs are provided 
in
Appendix~\ref{app:proofs}.

\section{Abstraction-Based Controller Design}

In this section, we review UMDP abstractions in a general way. We characterize the soundness and define the ambiguity diameter of a UMDP abstraction. Additionally, we state several standard definitions in abstraction-based controller design, such as the \emph{deterministic finite automaton} (DFA) of an \LTLf formula, the product UMDP, the product between $\Scal$ and the DFA and the corresponding Bellman operators.

\subsection{UMDP Abstraction}

We first define the set $\Xabs$, which contains $\Xsafe$ and also its possible successor states $x_{t+1}$ with probability $1$: $\Xabs := \Xsafe \cup \{f(x,u,w) \in X : x\in X_\safe, u\in U, w\in W\}$. Next, we partition this set into a finite set of regions that respect the regions of interest:
\begin{definition}
    We say that the partition $B$ of $\Xabs$ is $R$-respecting if, for each $b\in B$ and $r\in R$, either $b \subseteq r$ or $b \subseteq X \setminus r$. Furthermore, we let $\eta := (1/2)\max_{b\in B} \mathrm{diam}(b)$ be the \emph{size of the partition}, and we denote the partition by $B \equiv B^\eta$.
\end{definition}
Given an $R$-respecting partition $B^\eta$ of $\Xabs$, we pick, for each region $b \in B$, a representative point $s \in b$ such that its distance to the boundary of $b$ is strictly positive. We denote the set of representative points of $B^\eta$ by $S^\eta$. We let $S^\eta$ be the state space of the UMDP abstraction. This is slightly different from standard abstraction-based approaches, where the state space of the abstraction is simply $B^\eta$. The reason for this choice is that our approach leverages results of convergence of probability measures defined on a common space, i.e., $X$, which contains any $S^\eta$ defined as above, but not $B^\eta$. Figure~\ref{fig:sketch12} illustrates the partition process and the set $S^\eta$.
\begin{figure}[t]
    \centering
    \begin{subfigure}{0.49\linewidth}
        \centering
        \includegraphics[width=\linewidth]{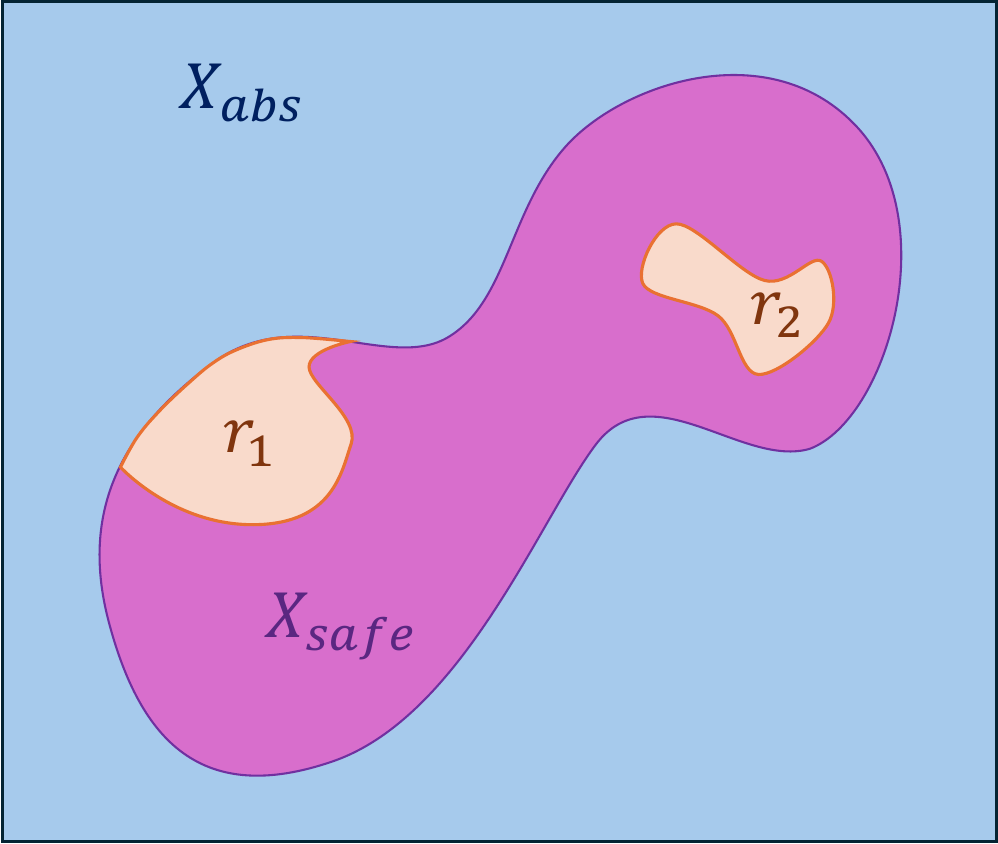}
        \caption{}
    \end{subfigure}
    \hfill
    \begin{subfigure}{0.49\linewidth}
        \centering
        \includegraphics[width=\linewidth]{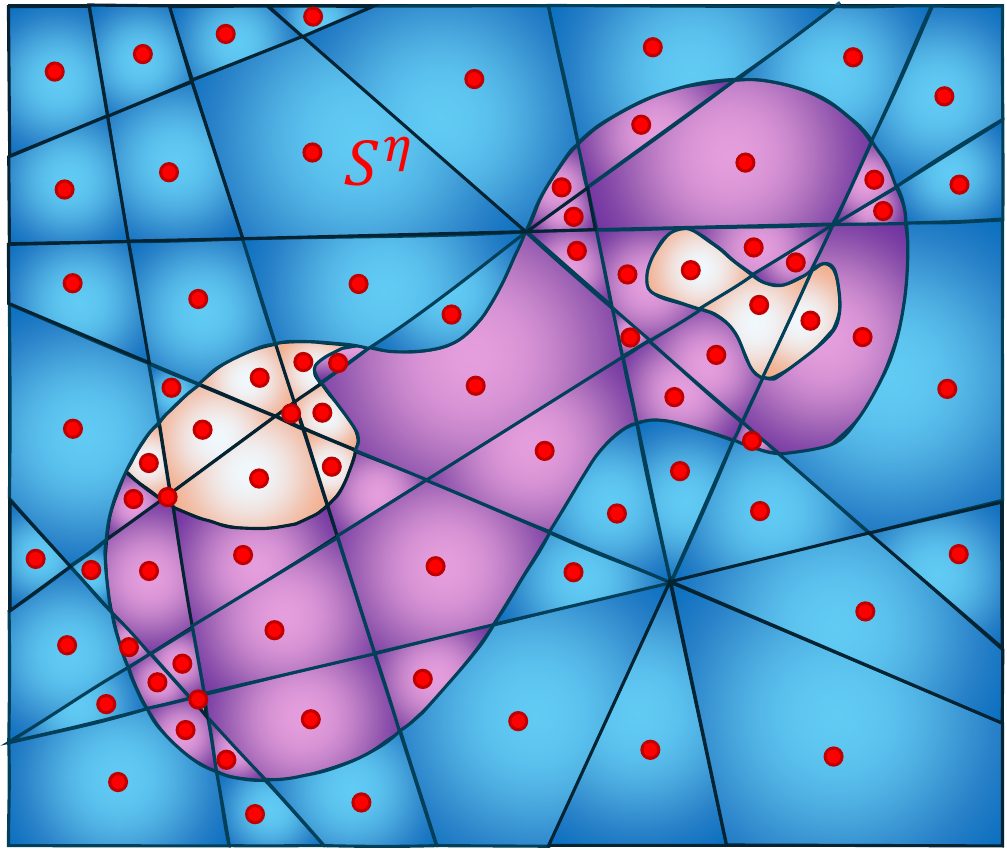}
        \caption{}
    \end{subfigure}
    \caption{Partition of $\Xabs$: (a) shows the continuous state-space with sets $\Xabs,\Xsafe$ and two regions of interest $r_1$ and $r_2$; (b) shows an $R$-respecting partition $B^\eta$, and the set of representative points $S^\eta$ in red.}
\label{fig:sketch12}
\end{figure}

The partition $B^\eta$ defines a \emph{refinement function} $\rho^\eta:\Xabs \to S^\eta$ that maps each continuous state $x\in\Xabs$ to the representative point $s\in S^\eta$ of the region $b\in B^\eta$ that contains $x$. For the sake of conciseness we also write $\rho^{-\eta}(x) := b$ for each $x\in \Xabs$, with $b\in B^\eta$ being the region $x$ belongs to. We also partition the disturbance set $W$ into the finite set of regions $C^\eta$, with size (diameter) at most $2\eta$. 
%

We now formally define a UMDP abstraction (without specifying how its transitions are obtained).
\begin{definition}[UMDP]
    Given an $R$-respecting partition $B^\eta$ of $\Xabs$, its set of representative points $S^\eta$, and the partition $C^\eta$ of $W$, a UMDP abstraction is a tuple $\Ucal^\eta := (S^\eta, A, \Gamma^\eta, \AP, L)$ where $S^\eta$ is the finite state space (with $S_\safe^\eta := S^\eta \cap \Xsafe$ denoting the safe states), $A := U$ is the finite set of actions, $\Gamma^\eta := \{\Gamma_{s,a}^\eta : s\in S^\eta, a\in A\}$ with $\Gamma_{s,a}^\eta \subseteq \Pcal(S^\eta )$\footnote{We assume closedness of the ambiguity sets $\Gamma_{s,a}^\eta$} being the set of transition probability distributions, or ambiguity set, of state-action pair $(s,a)$, and $L: S^\eta \to 2^{\AP}$ is the labeling function\footnote{More formally, it is the restriction of the labeling function $L$ of $\Scal$ to $S^\eta$. With a slight abuse of language we also denote it as $L$}.
\end{definition}
We define the \emph{ambiguity diameter} of a UMDP abstraction as a measure of the amount of ambiguity in the UMDP below, and as we show in Section~\ref{sec:algorithm}, this quantity plays a key role when it comes to completeness of our iterative algorithm.
\begin{definition}[Ambiguity Diameter]
\label{def:ambiguity_diameter}
    Let $\Ucal^\eta := (S^\eta, A, \Gamma^\eta, \AP, L)$ be a UMDP abstraction of system $\Scal$. 
    We denote the \emph{ambiguity diameter} of $\Ucal^\eta$ by $\phi(\Ucal^\eta) := \max_{s \in S^\eta_\safe, a\in A} \mathrm{diam}_{\W}( \Gamma_{s,a}^\eta)$, i.e., the maximum (Wasserstein) diameter of its ambiguity sets.
\end{definition}


Next, we discuss the semantics of UMDPs and provide the standard definitions of strategy and adversary. Given a UMDP $\Ucal^\eta$, we refer by \emph{path} of $\Ucal^\eta$ of length $T\in \naturals_0$ to a sequence $\omega_T^\eta = s_0 \xrightarrow{a_0}
\ldots \xrightarrow{a_{T-1}} s_T$ such that $s_t \in S^\eta$ for all $0 \leq t \leq T$, and $a_t \in A$ for all $0 \leq t \leq T-1$. We denote by $\Omega_T^\eta$ the sets of all such paths and let $\omega_T^\eta(t)$ denote the state of a path $\omega_T^\eta$ at time $t\in\{0, \dots, T\}$. Analogous to a controller of $\Scal$, a \emph{strategy} $(\sigma_t)_{t \in \{0,\dots, T-1\}}$ of $\Ucal^\eta$ is a sequence of functions where each $\sigma_t:\Omega_t^\eta\to A$ chooses the next action given the path followed up to time $t$. We denote by $\Sigma^\eta$ the set of all strategies of $\Ucal^\eta$. When each $\sigma_t$ only depends on $\omega_t^\eta(t)$, we say that $\sigma$ is \emph{Markovian} or \emph{memoryless}. Given a finite path of length $t$ $\omega_t^\eta \in \Omega_t^\eta$ with last state $s_t := \omega_t^\eta(t)$ and a strategy $\strategy \in \Sigma^\eta$, $\Ucal^\eta$ transitions from $s_{t}$ under $a_t := \sigma_t(\omega_t^\eta)$ to $s_{t+1}$ according to some probability distribution in $\Gamma_{s_t,a_t}^\eta$, which is chosen by the adversary \cite{givan2000bounded,wolff2012robust}. Formally speaking, an \emph{adversary} is a function $\xi: S^\eta \times A \times \naturals_0 \rightarrow \mathcal{P}(S^\eta)$ that chooses, given current state $s_t$, action $a_t$ and time step $t$, a probability distribution $\gamma\in\Gamma_{s_t,a_t}^\eta$, according to which $s_{t+1}$ is distributed. We denote the set of all adversaries of $\Ucal^\eta$ by $\Xi^\eta$. Given an initial state $s_0\in S^\eta$, a strategy $\sigma\in\Sigma^\eta$, and an adversary $\xi\in\Xi^\eta$, $\Ucal^\eta$ collapses to a Markov chain, with a unique probability measure over its finite paths. We denote this measure by $\text{Pr}_{\eta, s_0}^{\sigma,\xi}$.

Now we formalize the notion that the UMDP abstraction captures all behaviors of interest of $\Scal$. Such abstractions ensure that the obtained controllers and satisfaction guarantees correctly translate to the original system $\Scal$.
\begin{definition}[Soundness]
\label{dfn:soundness}
    Consider a given System~$\Scal$ and its UMDP abstraction $\Ucal^\eta := (S^\eta, A, \Gamma^\eta, \AP, L)$. We say that $\Ucal^\eta$ is a \emph{sound abstraction} of $\Scal$ if (i) for each $x\in \Xsafe$ with $\rho^\eta(x) = s$ and $a\in A$, the distribution $\Tau^\eta(\cdot \mid x,a) \in \Pcal(S^\eta)$, where $\Tau^\eta(\rho^{-\eta}(s') \mid x,a) := T(\rho^{-\eta}(s') \mid x,a)$ for all $s'\in S^\eta$, satisfies that $\Tau^\eta(\cdot \mid x,a) \in \Gamma_{s,a}^\eta$, and (ii) for each $x\notin \Xsafe$ with $\rho^\eta(x) = s$ and $a\in A$, the ambiguity set is $\Gamma_{s,a}^\eta = \{\delta_s\}$.
\end{definition}

Our approach assumes a given procedure of obtaining the abstraction. In consequence, the sufficient condition for completeness of the algorithm is a condition on the abstraction procedure. In order to keep our theoretical analysis general, we now formalize what we mean by abstraction procedure:
\begin{definition}[Abstraction Procedure]
    Given a UMDP abstraction class $\mathbb{U}$, an \LTLf formula $\varphi$ and $\eta > 0$, we define an \emph{abstraction procedure} $\mathrm{Abs}_\mathbb{U}$ as the mapping $(\mathcal{S}, \eta) \mapsto \mathrm{Abs}_\mathbb{U}(\mathcal{S}, \eta)$, which soundly abstracts System~$\mathcal{S}$ into a UMDP $\mathcal{U}^\eta := (S^\eta,A,\Gamma^\eta, \AP, L) \in\mathbb{U}$, where the size parameter of the partitions $S^\eta$ and $C^\eta$ is at most $\eta$.
\end{definition}

\subsection{Robust Dynamic Programming}

Here, we describe the standard methodology for control synthesis given a sound UMDP abstraction $\Ucal$. Leveraging the fact that for each \LTLf formula, a DFA can be constructed that precisely accepts the language of the formula \cite{Giacomo2013}, synthesis consists of four main steps: (1) obtaining the \emph{deterministic finite automaton} (DFA) $\Acal^\varphi$ that represents $\varphi$, (2) obtaining the \emph{product UMDP}, i.e., a different UMDP $\Ucal^\varphi := \Ucal\otimes \Acal^\varphi$ defined as the product between $\Ucal$ and $\Acal^\varphi$, (3) performing \emph{robust dynamic programming} (RDP) on $\Ucal^\varphi$ with a reachability objective, which yields a strategy and bounds in the satisfaction probability for the product, and (4) correctly translating these results to $\Scal$.

We define the DFA of $\varphi$ below.
\begin{definition}
    Let $\varphi$ be an \LTLf formula defined over $\AP$. The deterministic finite automaton (DFA) constructed from $\varphi$ is a tuple $\Acal^\varphi = (Z, z_\mathrm{init}, 2^\AP, \delta, Z_\text{acc})$, where $Z$ is the finite set of states, $z_\mathrm{init}\in Z$ is the initial state, $2^\AP$ is the alphabet (symbols), $\delta:Z\times 2^\AP \rightarrow Z$ is the deterministic transition function, and $Z_\text{acc} \subseteq Z$ is the set of accepting states.
\end{definition}
Given a trace $l = l_0l_1 \dots l_T \in (2^{AP})^T$, a run $z = z_0z_1 \dots z_{T+1}$ is induced on $\A^\varphi$, where $z_0 = z_\mathrm{init}$ and $z_{t+1} = \delta(z_t,l_t)$ for all $t\in \{0,\dots,T\}$. By construction of $\A^\varphi$, trace $l$ satisfies $\varphi$ iff it is \emph{accepting}, i.e., iff $z_{T+1} \in Z_\mathrm{acc}$ \cite{de2013linear}. Next, we define the product between System~$\Scal$ and $\Acal^\varphi$, whose transitions carry information about both the dynamics of $\Scal$ and satisfaction of $\varphi$.
\begin{definition}[Product MDP]
Let $\Acal^\varphi = (Z, z_\mathrm{init}, 2^\AP, \delta, Z_\text{acc})$ be the DFA corresponding to $\varphi$. The product $\Scal ^\varphi := \Scal \otimes \Acal^\varphi$ is an MDP $\Scal^\varphi := (X^\varphi, A, \Tau^\varphi, \AP, L^\varphi)$ 
over the hybrid space $X^\varphi = X\times Z$, with actions $A := U$, labeling function $L^\varphi(x,z) := L(x)$ for all $(x,z) \in X^\varphi$ and transition kernel $\Tau^\varphi$ defined as
\begin{align*}
    \Tau^\varphi(D \times \{z'\} \mid (x, z)) = 
        \Tau(\{x'\in D : z' = \delta(z, L(x'))\}),
\end{align*}
for all $x \in X$, $D \in \Fcal(X)$, $z, z' \in Z$.
\end{definition}
We note that, since we consider specifications that require the system to remain in $\Xsafe$, $\A^\varphi$ always has an absorbing state $z_\mathrm{unsafe} \in Z$ that is reached when the label of the successor state $x'$ does not contain $\Prop_\mathrm{safe}$. Once a trajectory of $\Scal^\varphi$ reaches this state, $\varphi$ cannot be satisfied. We next define the product UMDP.
\begin{definition}[Product UMDP]
    Given a UMDP $\Ucal$ and the DFA $\Acal^\varphi$, we define the \emph{product UMDP} $\Ucal^\varphi := \Ucal \otimes \Acal^\varphi$ 
    as the UMDP $\Ucal^\varphi := (S^\varphi, A, \Gamma^\varphi, \AP,L^\varphi)$, where $S^\varphi := S\times Z$, $L^\varphi(x,z) = L(x)$ for all $X \in X$ and $z\in Z$, and $\Gamma^\varphi := \{\Gamma_{(s,z),a}^\varphi : (s,z)\in S^\varphi, a\in A\} \subset \Pcal(S^\varphi)$, with%
    \begin{multline}
        \Gamma_{(s,z),a}^\varphi := \bigg\{\gamma^\varphi \in \Pcal(S^\varphi) :\\
        \gamma^\varphi(s',\delta(z,L(s'))) = \gamma(s')\:,\:\gamma\in \Gamma_{s,a}, s' \in S \bigg\},
    \end{multline}
    for all $X \in X$ and $z\in Z$.
\end{definition}
The products $\Scal^\varphi$ and $\Ucal^\varphi$ enable characterizing optimal controllers and satisfaction probabilities through RDP, which boils down to probabilistic reachability computations on the products. We next define the \emph{Bellman operator} of the product $\Scal^\varphi$, which appears in the RDP iterations:
\begin{definition}[Bellman Operator]
    Consider the product MDP $\Scal^\varphi$ of $\Scal$ with the DFA $\Acal^\varphi$. Denote by $\Vcal^\varphi$ the set of value functions on $X^\varphi$, i.e., $\Vcal^\varphi := \{v : X^\varphi \rightarrow [0,1]\}$. Given a value function $v$ the Bellman operator $\Bcal:\Vcal^\varphi \rightarrow \Vcal^\varphi$ of $\Scal^\varphi$ is defined as $\Bcal[v](x,z) = \max_{a\in A} \Bcal_a[v](x,z)$, with
\begin{align*}
    &\Bcal_a[v](x,z) := \mathds{1}_{Z_\text{acc}}(z) + \mathds{1}_{Z\setminus Z_\text{acc}}(z)\int_{X^\varphi} v d\Tau^\varphi(\cdot \mid (x, z), a)\\
    &= \mathds{1}_{Z_\text{acc}}(z) + \mathds{1}_{Z\setminus Z_\text{acc}}(z)\int_X v(x', \delta(z, L(x'))) \Tau(dx' \mid x, a),
\end{align*}
for all $x\in X$ and $z\in Z$.
\end{definition}

In analogy with the previous definition, we now define the Bellman operators of the product UMDP $\Ucal^{\eta, \varphi}$.
\begin{definition}[Ambiguous Bellman Operators]
    Consider a UMDP abstraction $\Ucal^\eta$ of $\Scal$ and its product $\Ucal^\varphi$ with $\Acal^\varphi$. For each $v \in \Vcal^\varphi$ and $z\in Z$, denote by $v^\eta \in \Vcal^\varphi$ the tightest piecewise constant minorizer of $v$ on the partition $B^\eta$, i.e., the value function satisfying $v^\eta(x,z) := \inf\{v(y,z) : y \in \Xabs, \rho^\eta(x) = \rho^\eta(y)\}$ for $x\in \Xabs$ and $v^\eta(x,z) = 0$ elsewhere. We denote by $\underline \Bcal^{\eta}, \overline \Bcal^{\eta}:\Vcal^\varphi\rightarrow\Vcal^\varphi$ (respectively) the \emph{pessimistic} and \emph{optimistic} Bellman operators, defined as $\underline \Bcal^{\eta}[\cdot](x,z) := \max_{a\in A} \underline \Bcal_a^{\eta}[\cdot](x,z)$ and $\overline \Bcal^{\eta}[\cdot](x,z) := \max_{a\in A} \overline \Bcal_a^{\eta}[\cdot](x,z)$ for all $x\in \Xsafe$ and $z\in Z$, with
    \begin{align*}
        &\underline \Bcal_a^{\eta}[v](x,z) := \mathds{1}_{Z_\text{acc}}(z) + \mathds{1}_{Z\setminus Z_\text{acc}}(z) \min_{\gamma \in\Gamma_{(s,z),a}^{\eta,\varphi}} \int_{X^\varphi} v^\eta d\gamma\\
        &= \mathds{1}_{Z_\text{acc}}(z) + \mathds{1}_{Z\setminus Z_\text{acc}}(z) \min_{\gamma\in\Gamma_{s,a}^{\eta}} \int_{X} v^\eta(\cdot,\delta(z,L(\cdot))) d\gamma,
    \end{align*}
    and
    \begin{align*}
        &\overline \Bcal_a^{\eta}[v](x,z) := \mathds{1}_{Z_\text{acc}}(z) + \mathds{1}_{Z\setminus Z_\text{acc}}(z) \max_{\gamma \in\Gamma_{(s,z),a}^{\eta,\varphi}} \int_{X^\varphi} v^\eta d\gamma\\
        &= \mathds{1}_{Z_\text{acc}}(z) + \mathds{1}_{Z\setminus Z_\text{acc}}(z) \max_{\gamma\in\Gamma_{s,a}^{\eta}}\int_{X} v^\eta(\cdot,\delta(z,L(\cdot))) d\gamma
    \end{align*}
    for each $x \in \Xsafe$ with $s = \rho(x)$, $z\in Z$ and $a\in A$. We also define $\underline \Bcal^\eta[v](x,z) = \overline \Bcal^\eta[v](x,z) := 0$ for all $x\notin \Xsafe$ and $z\in Z$.
    %

    %
\end{definition}

Given the products $\Scal^\varphi$ and $\Ucal^{\eta,\varphi}$, we next describe how to run RDP  on them. While RDP on $\Ucal^{\eta,\varphi}$ is the actual algorithm that is implemented, yielding a controller and bounds in the probability of satisfying $\varphi$, RDP on $\Scal^\varphi$ is a mere mathematical tool that enables reasoning about the correctness of the obtained controller and probability bounds.
\begin{proposition}[RDP]
\label{prop:RDP}
    Consider system $\mathcal{S}$ and its UMDP abstraction $\Ucal^\eta := (S^\eta, A, \Gamma^\eta, \AP, L)$, as well as the products $\Scal^\varphi$ and $\Ucal^{\eta,\varphi}$ with the DFA $\Acal^\varphi$. Define the sequence of \emph{pessimistic} value functions $(\underline v_k)_{k \le T} \subset \Vcal^\varphi$ satisfying $\underline v_0 = \mathds{1}_{X_\text{acc}^\varphi}$ and $\underline v_{k+1} = \underline \Bcal^{\eta}[\underline v_k]$, for all $k \le T$. Define $\underline p^\eta(\cdot) = \underline v_T(\cdot, \delta(z_\mathrm{init}, L(\cdot)))$. Then \cite{gracia2025efficient},
    \begin{itemize}
        \item $\underline p^\eta(x_0) = \sup_{\sigma\in\Sigma^{\eta, \varphi}} \inf_{\xi\in \Xi^{\eta, \varphi}} \mathrm{Pr}_{\eta,s_0}^{\sigma,\xi}[\varphi, T]$, with $s_0 = \rho(x_0)$, 
        for all $x_0 \in \Xabs$.
    \end{itemize}
    Similarly, let $(\overline v_k)_{k \le T} \subset \Vcal^\varphi$ denote the \emph{optimistic} sequence satisfying $\overline v_0 = \mathds{1}_{X_\text{acc}^\varphi}$ and $\overline v_{k+1} = \overline \Bcal^{\eta}[\overline v_k]$, for all $k \le T$. Define $\overline p^\eta(\cdot) = \overline v_T(\cdot, \delta(z_\mathrm{init}, L(\cdot)))$. Then \cite{gracia2025efficient},
    \begin{itemize}
        \item $\overline p^\eta(x_0) = \sup_{\sigma\in\Sigma^{\eta, \varphi}} \sup_{\xi\in \Xi^{\eta, \varphi}} \mathrm{Pr}_{\eta,s_0}^{\sigma,\xi}[\varphi, T]$, with $s_0 = \rho(x_0)$,
        for all $x_0 \in \Xabs$.
    \end{itemize}
    Define the sequence of value functions $(v_k)_{k \le T} \subset \Vcal^\varphi$ satisfying $v_0 = \mathds{1}_{X_\text{acc}^\varphi}$ and $v_{k+1} = \Bcal[v_k]$ for all $k \le T$. Let $p(\cdot) = v_T(\cdot, \delta(z_\mathrm{init}, L(\cdot)))$. Then \cite{bertsekas1996stochastic},
    \begin{itemize}
        \item $p(x_0) = \mathrm{Pr}_{x_0}^*[\varphi, T]$ for all $x_0 \in X$.
    \end{itemize}
    Define a Markovian strategy $\sigma^{\eta, \varphi}\in \Sigma^{\eta, \varphi}$ such that $\sigma_{t}^{\eta, \varphi}(s,z) \in \arg\max_{a\in A} \underline \Bcal_a^{\eta}[\underline v_{T-t-1}^\eta](s,z)$ for each $t\in \{0,\dots, T-1\}$, $s\in S^\eta$ and $z\in Z$. Map $\sigma^{\eta, \varphi}$ to the controller $\kappa^\eta$ of $\Scal$ by defining $\kappa^\eta(\omega_t^X) := \sigma^{\eta,\varphi}(\omega_t^{\eta,\varphi}(t))$, where $\omega_t^{\eta,\varphi} \in \Omega_t^{\eta,\varphi}$ is the path of $\Ucal^{\eta,\varphi}$ generated by path $\omega_t^X \in \Omega_t^X$ of $\Scal$. Then \cite{skovbekk2021formal}, for all $x_0 \in \Xabs$
    \begin{itemize}
        \item $\mathrm{Pr}_{x_0}^{\kappa^\eta}[\varphi, T], \mathrm{Pr}_{x_0}^{*}[\varphi, T] \in [\underline p^\eta(x_0), \overline p^\eta(x_0)]$.
    \end{itemize}
\end{proposition}
%

%

Figure~\ref{fig:sketch3} illustrates some of the concepts in Proposition~\ref{prop:RDP}. Specifically, it shows the sequences of satisfaction probability functions obtained via RDP on $\Scal^\varphi$ and on the abstraction $\Ucal^{\eta,\varphi}$ for two different partition sizes.
\begin{figure}
    \centering
    \includegraphics[width=0.8\linewidth]{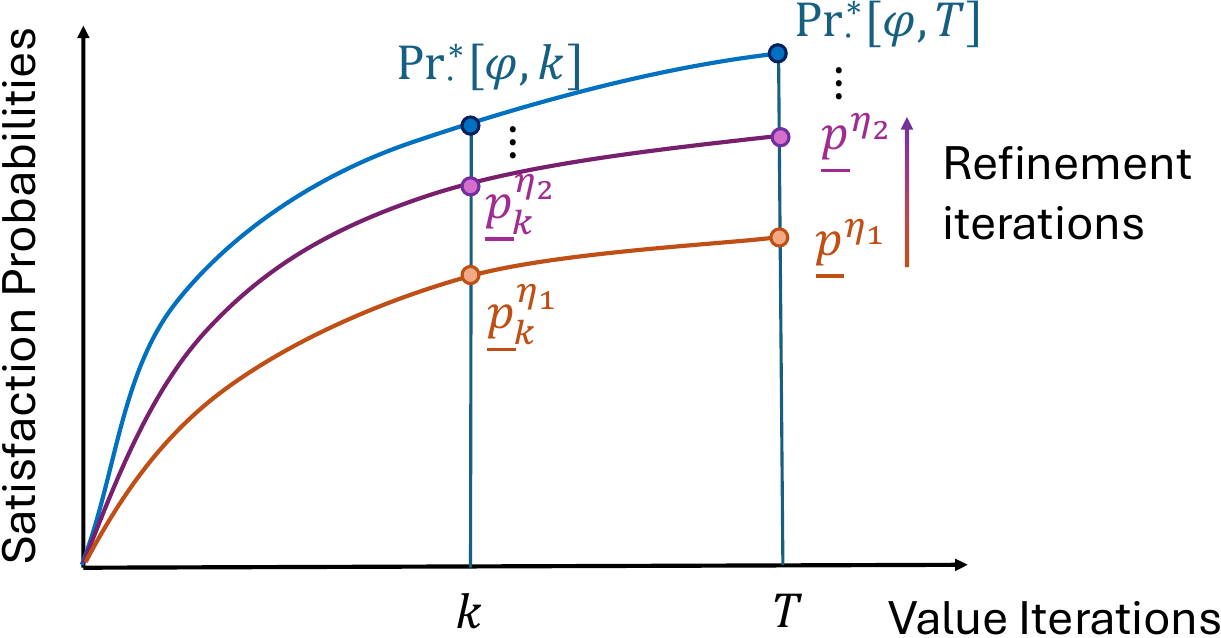}
    \caption{Graphical illustration of Proposition~\ref{prop:RDP} and Definition~\ref{def:asymptotic_optimality}. Value functions v.s. (RDP) iteration number for two discretization sizes $\eta_1 <\eta_2$. We let $\underline p_k^\eta := \underline u_k^\eta(\cdot,\delta(z_\text{init}, L(\cdot)))$. If the abstraction procedure is asymptotically optimal, then the orange and pink curves converge to the blue one.}
    \label{fig:sketch3}
\end{figure}
As previously stated, we desire an abstraction procedure that is asymptotically optimal, i.e., one for which refinement yields, in the limit, an optimal controller and zero gap in the satisfaction probability bounds. We formally define asymptotic optimality in Definition~\ref{def:asymptotic_optimality}, and illustrate it in Figure~\ref{fig:sketch3}.
\begin{definition}[Asymptotic Optimality]
\label{def:asymptotic_optimality}
    Consider the abstraction procedure $\mathrm{Abs}_\mathbb{U}$ and let $(\underline p^\eta, \overline p^\eta, \kappa^\eta)$ be the result of running RDP on $\Ucal^\eta = \mathrm{Abs}_\mathbb{U}(\mathcal{S}, \eta)$, for $\eta > 0$. We say that $\mathrm{Abs}_\mathbb{U}$ is asymptotically optimal if for each \LTLf formula $\varphi$ and horizon $T \in\naturals_0$ it holds that $P_{\cdot}^{\kappa^\eta}[\varphi, T]$, $\underline p_T^\eta$, and $\overline p_T^\eta$ uniformly converge to $P_{\cdot}^*[\varphi, T]$ as $\eta\to 0$. 
\end{definition}

\section{Algorithm for Near-Optimal Control}
\label{sec:algorithm}

In this section we present Algorithm~\ref{alg:algo} to solve Problem~\ref{prob:prob}, and a sufficient condition, namely, vanishing ambiguity, on the abstraction procedure that ensures completeness of the algorithm, i.e., that it successfully terminates in finite time. Our algorithm soundly abstracts $\Scal$ into a UMDP abstraction for a partition that is iteratively refined. At each iteration, the algorithm obtains a strategy for the abstraction and bounds in the probability of satisfying the specification through RDP on the abstraction. If the gap between the bounds is smaller than some user-defined $\varepsilon>0$, then the algorithm stops, ensuring that the synthesized controller is $\varepsilon$-optimal, i.e., that it solves Problem~\ref{prob:prob}.
\begin{algorithm}[t]\small
\caption{$\varepsilon$-Optimal Control Synthesis}
\label{alg:algo}
\begin{algorithmic}[1]
    \Require $\mathcal{S}, \varphi, T, \varepsilon, \mathrm{Abs}_{\mathbb{U}}, (\eta_i)_{i\in \naturals}$ such that $\eta_i > 0$, $\eta_i > \eta_{i+1}$ and $\lim_{i\to \infty} \eta_i = 0$
    \Ensure $\epsilon$-optimal controller $\kappa^\epsilon$, $\underline p^\epsilon, \overline p^\epsilon$
    \State Compute $\Acal$ from $\varphi$
    \State $i \gets 0$
    \While{$\max\{\overline p(x) - \underline p(x) \mid x\in X\} > \epsilon$}
    \State $i \gets i+1$
    \State $\Ucal \gets \mathrm{Abs}_{\mathbb{U}}(\mathcal{S}, \eta_i)$
    \State Obtain $\underline p, \overline p$ and $\sigma^{\varphi}$ via RDP on $\Ucal\otimes \Acal_\varphi$
    \EndWhile
    \State $\kappa^\epsilon \gets$ refine $\sigma^\varphi$ to $\Scal$
    \State $(\underline p^\epsilon, \overline p^\epsilon) \gets$ $(\underline p, \overline p)$
\end{algorithmic}
\end{algorithm}

We now state 
five lemmas
which enable us to find a sufficient condition for asymptotic optimality of an abstraction class. Proofs are provided in
Appendix~\ref{app:proofs}. Lemmas~\ref{lemma:uc_value_functions} through \ref{lemma:aux_3_term} allow us to bound the $1$-step difference in the value functions of $\Scal^\varphi$ and $\Ucal^{\eta,\varphi}$ in Lemma~\ref{lemma:uc_operator}. We then extend this result to multiple time-steps in Theorem~\ref{thm:asymptotic_optimality}. The following lemma states that, when the current state $x_t$ of $\Scal$ is perturbed to some $y$, the magnitude of the perturbation controls the variation in the expected value of uniformly continuous functions at $x_{t+1}$.
\begin{lemma}
\label{lemma:uniform_convergence_kernel}
    Let $v\in \Vcal := \{v : X \rightarrow [0,1]\}$ be uniformly continuous with $\mu:[0,\infty) \rightarrow [0, \infty)$ being a concave and non-decreasing modulus of continuity of $v$. Then,
    \begin{multline*}
        \left| \int_X v  d\Tau(\cdot\mid x, a) - \int_X v  d\Tau(\cdot\mid y, a) \right|\\
        \le \mu(\W(\Tau(\cdot\mid x, a), \Tau(\cdot\mid y, a))) \le \mu(L_f\|x-y\|),
    \end{multline*}
    for all $x,y\in X$ and $a\in A$. 
\end{lemma}
Lemma~\ref{lemma:uc_value_functions} establishes uniform continuity of the value functions of $\Scal^\varphi$.
\begin{lemma}
    \label{lemma:uc_value_functions}
    Consider the value functions $(v_k)_{k \le T}$ of $\Scal^\varphi$ in Prop.~\ref{prop:RDP}. Then, each function $v_k(\cdot,z)$ are uniformly continuous on $X$ for all $z\in Z$.
\end{lemma}
The following lemma compares the expected value of uniformly continuous functions evaluated at the successor state of $\Ucal^{\eta,\varphi}$, when said successor is distributed according to $\Tau^\eta$ and to some distribution in $\Gamma_{s,a}^\eta$, respectively. The lemma 
shows that fine partitions and a vanishing ambiguity diameter guarantee that the ambiguity in the transitions of $\Ucal^\eta$ has small effect on said expected value. More importantly, it guarantees that the convergence is uniform in $x\in \Xsafe$.
%
\begin{lemma}
\label{lemma:aux_2_term}
    Let $\lim_{\eta\to 0} \phi(\mathrm{Abs}_{\mathbb{U}}(\mathcal{S},\eta)) = 0$, where $\phi(\mathrm{Abs}_{\mathbb{U}}(\mathcal{S},\eta))$ is the ambiguity diameter of the UMDP $\Ucal^\eta = \mathrm{Abs}_{\mathbb{U}}(\mathcal{S},\eta)$ (see Definition~\ref{def:ambiguity_diameter}). Let $v \in \Vcal^\varphi$, with $v(\cdot, z')$ being uniformly continuous on $X$ for all $z\in Z$, and define $v_{z} := v(\cdot, \delta(z,L(\cdot)))$ for all $z\in Z$. 
    Further,
    for each $x\in \Xsafe$, $a\in A$ and $\eta > 0$, let $\gamma_{x,a}^\eta \in\Gamma_{s,a}^\eta$, with $s = \rho(x)$
    . Then it holds that, for each $z\in Z$ and $a\in A$,
    \begin{align*}
        \lim_{\eta\to 0} \left| \int_X v_z d\gamma_{x,a}^\eta - \int_X v_z d\Tau^\eta(\cdot \mid x, a) \right| = 0
    \end{align*}
    uniformly on $x\in \Xsafe$.
\end{lemma}
The following lemma gives a similar uniform convergence result as that of Lemma~\ref{lemma:aux_2_term}, but instead considers expectations taken with respect to $\Tau^\eta$ and $\Tau$. Intuitively, this quantifies the error introduced by discretizing the transition kernel.
\begin{lemma}
\label{lemma:aux_3_term}
    Let 
    $v_z$ be as defined in Lemma~\ref{lemma:aux_2_term} for each $z\in Z$. 
    Then, it holds that, for each $z\in Z$ and $a\in A$,
    \begin{align*}
        \lim_{\eta\to 0} \left| \int_X v_z d\Tau^\eta(\cdot \mid x, a) - \int_X v_z d\Tau(\cdot \mid x, a) \right| = 0
    \end{align*}
    uniformly on $x\in \Xsafe$.
\end{lemma}
The next lemma shows that fine partitions and a vanishing ambiguity diameter guarantee that the Bellman operators $\Bcal$ and $\Bcal^\eta$ yield similar outputs for uniformly continuous input functions that are identically zero at $z_\mathrm{unsafe}$. The lemma also shows that the difference in said outputs vanishes, uniformly on $x\in \Xsafe$, as the partitions are refined.
\begin{lemma}
\label{lemma:uc_operator}
    Let  $v\in\Vcal^\varphi$ with $v(\cdot, z_\text{unsafe}) = 0$, and $v(\cdot,z)$ be uniformly continuous on $X$ for all $z\in Z$. Assume that $\lim_{\eta\to0} \phi(\mathrm{Abs}(\Scal,\eta)) = 0$. 
    Then, for each $z\in Z$,
    \begin{align*}
        \lim_{\eta\to0} |\underline \Bcal^\eta[v](x,z) - \Bcal[v](x,z)| = 0
    \end{align*}
    uniformly on $x\in \Xsafe$.
\end{lemma}
Intuitively, Lemma~\ref{lemma:uc_operator} quantifies the $1$-step error due to the conservatism in the abstraction, showing that it vanishes as the partitions are refined. 
%
%
By extending this result to multiple time steps, we are able to provide formal statements on optimality and completeness.
\begin{theorem}[Asymptotic Optimality]
\label{thm:asymptotic_optimality}
    Consider given system $\mathcal{S}$ and abstraction procedure $\mathrm{Abs}_\mathbb{U}$. Then, $\mathrm{Abs}_\mathbb{U}$ is asymptotically optimal if $\lim_{\eta\to 0}\phi(\mathrm{Abs}_\mathbb{U}(\mathcal{S}, \eta)) = 0$, i.e., if the ambiguity in the UMDP $\mathcal{U}^\eta = \mathrm{Abs}(\mathcal{S}, \eta)$ vanishes with $\eta$.
\end{theorem}
The following corollary is a direct result of Theorem~\ref{thm:asymptotic_optimality}.
\begin{corollary}[Completeness of Algorithm~\ref{alg:algo}]
\label{cor:near_optimal}
    Let $\mathcal{S}$ and $\mathrm{Abs}_\mathbb{U}$ be given. If $\mathrm{Abs}_\mathbb{U}$ is asymptotically optimal, then 
    Algorithm~\ref{alg:algo} solves Problem~\ref{prob:prob} in \emph{finite time} for any \LTLf formula $\varphi$, horizon $T\in \naturals_0$, and $\varepsilon > 0$.
\end{corollary}
We highlight the generality of Theorem~\ref{thm:asymptotic_optimality} and Corollary~\ref{cor:near_optimal}, as they do not impose restrictive assumptions on the partitions or discretization scheme: only that the partitions must be $R$-respecting and that $(\eta_i)_{i\in \naturals}$ must decrease monotonously to $0$
(see Algorithm~\ref{alg:algo}), i.e., the size of each region in the partition vanishes in the limit. For example, one could choose the shape of $S^{\eta_i}$ in an adaptive way, namely, depending on $S^{\eta_{i-1}}$ and on the results of iteration $i-1$.
\begin{remark}
    %
    Determining whether or not the \emph{sufficient} condition in Theorem~\ref{thm:asymptotic_optimality} is also \emph{necessary} for asymptotic optimality of Algorithm~\ref{alg:algo} is not straightforward, as our proofs employ several overapproximations and. Furthermore, our definition of asymptotic optimality in Definition~\ref{def:asymptotic_optimality} requires the abstraction procedure to yield, in the limit as $\eta \to 0$ and via RDP, optimal results for every \LTLf property. In consequence, we do not rule out the possibility that there may exist some \LTLf property and an abstraction procedure that yield an optimal controller as $\eta \to 0$ and/or that Algorithm~\ref{alg:algo} returns a $\varepsilon$-optimal controller for some $\varepsilon > 0$. However, we highlight the importance of designing an abstraction procedure that is guaranteed to provide $\varepsilon$-optimal results for every \LTLf formula and $\varepsilon>0$. Otherwise, it is not clear if a non-asymptotically optimal procedure will work for the formula and $\varepsilon$ of interest.
\end{remark}

\section{Common UMDP Abstractions}

In this section, we examine common UMDP abstraction procedures, namely, SMDP and IMDP abstraction via reachability computations, and study their behavior as the partition is refined. We show that, while these SMDPs are asymptotically optimal, IMDPs (obtained as per \cite{10347405}) are not. We highlight that our conclusions apply to the specific procedure of obtaining IMDP and SMDP abstractions, and provide no claim towards optimality of the abstraction class, as different abstraction procedures may lead to different optimality features. We first define the $1$-step reachable set of regions $b \subseteq X$ and $c \subseteq W$ under action $a\in A$ as the set
\begin{align*}
    \mathcal{R}(b,a,c) := \{f(x,a,w) \in X : x \in b, w \in c\}.
\end{align*}

\subsection{SMDPs}
An SMDP $\Ucal^{\mathbb{S}}$ is a UMDP class whose ambiguity set $\Gamma^{\mathbb{S}}$ has the following structure \cite{gracia2025beyond}: each probability distribution $\gamma \in \Gamma^{\mathbb{S}}$ corresponds to transitioning to a set of states (or \emph{cluster}) with a given probability, and the actual successor state inside that cluster being distributed according to some conditional probability distribution. In consequence, given a conditional probability distribution for each cluster, a distribution $\gamma \in \Gamma^{\mathbb{S}}$ is identified. Therefore, the ambiguity in the transitions of $\Ucal^{\mathbb{S}}$ comes from allowing the conditional distributions to take any value within the corresponding probability simplices. In the following definition we formalize this intuition and show how to compute the clusters of the SMDP abstraction.
\begin{definition}[SMDP Abstraction \cite{gracia2025beyond}]
    \label{def:mdp_st_abstraction}
    Consider System $\Scal$, the $R$-respecting partition $B^\eta$ and the partition $C^\eta$ of $W$. For all $s\in S_\safe^\eta$, $a\in A$, $c\in C^\eta$, define cluster $q_{s,a,c} := S^\eta\cap \mathcal{R}(\rho^{-1}(s),a,c)$, and let $Q_{s,a}^\eta := \{q_{s,a,c}: c\in C^\eta\}$. For each $q_{s,a,c} \in Q_{s,a}^\eta$, let $\theta_{s,a,c} \in \Pcal(q_{s,a,c})$. For each 
    $$\theta_{s,a} := \big(\theta_{s,a,c} \big)_{c \in C^\eta} \in \bigtimes_{c\in C^\eta} \mathcal{P}(q_{s,a,c}) =: \Theta_{s,a}^\eta,$$ denote $\gamma_{\theta_{s,a}} \in \mathcal{P}(S^\eta)$ the probability distribution such that
    \begin{align}
        \label{eq: SMDP exact successor prob}
        \gamma_{\theta_{s,a}}(s') = \sum_{c \in C^\eta : s' \in q_{s,a,c}} \theta_{s,a,c}(s') P_W(c),
    \end{align}
    for $s' \in S^\eta$. We define the SMDP abstraction of System $\Scal$ as $\Ucal^{\mathbb{S}, \eta} := (S^\eta,A,\Gamma^{\mathbb{S}, \eta}, \AP, L)$, where
    \begin{align}
    \label{eq:Gamma_mdp_st}
    \Gamma_{s,a}^{\mathbb{S}, \eta} := \{\gamma_{\theta_{s,a}} : \theta_{s,a} \in \Theta_{s,a}\}
    \qquad \forall s\in S_\safe^\eta,\;a\in A,
    \end{align}
    and $\Gamma^{\mathbb{S}, \eta}_{s,a} := \{\delta_{s}\}$ for all $s\notin S_\safe^\eta$ and $a \in A$.
\end{definition}
As shown in \cite{gracia2025beyond}, such an SMDP abstraction is Sound as per Definition~\ref{dfn:soundness}. The following theorem establishes optimality of such an abstraction process.
\begin{theorem}
\label{thm:optimality_smdp}
    Let $\mathrm{Abs}_\mathbb{S}$ be the abstraction procedure that abstracts $\Scal$ into an SMDP as per Definition~\ref{def:mdp_st_abstraction}. Then, 
    \begin{itemize}
    
        \item $\lim_{\eta \to 0} \phi(\mathrm{Abs}_\mathbb{S}(\Scal,\eta)) = 0$,
        \item $\mathrm{Abs}_\mathbb{S}$ is asymptotically optimal,
        \item Algorithm~\ref{alg:algo} and $\mathrm{Abs}_\mathbb{S}$ solve Problem~\ref{prob:prob} in finite time.
    \end{itemize}
\end{theorem}
The proof of Theorem~\ref{thm:optimality_smdp} first shows that the first item in the theorem holds. This implies that Theorem~\ref{thm:asymptotic_optimality} and Corollary~\ref{cor:near_optimal} also hold, proving the second and third items. In order to show that $\phi(\mathrm{Abs}_\mathbb{S}(\Scal,\eta))\to 0$, the proof first bounds the Wasserstein distance between conditional distributions on each cluster $q_{s,a,c}$  of $\Ucal^{\mathbb{S},\eta} = \mathrm{Abs}_\mathbb{S}(\Scal,\eta)$. The bounds for all clusters are then combined to obtain a bound on the ambiguity diameter $\phi(\mathrm{Abs}_\mathbb{S}(\Scal,\eta))$, which is then shown to vanish with $\eta$, as the diameters of the clusters $q_{s,a,c}$ shrink when the partitions $S^\eta$ and $C^\eta$ are refined.

\subsubsection*{Smaller SMDP Abstraction}

Here we show that, when abstracting System~$\Scal$ to an SMDP, there is no need to partition the entirety of set $\Xabs$, but only the safe set $\Xsafe$. The key is to treat $\Xunsafe$ as a single unsafe set in $B^\eta$, and letting its representative point be some $s_\unsafe \in \Xunsafe$. This reduces the size of the abstraction, alleviating memory and computational complexity. The following theorem shows that performing RDP on such an SMDP yields the same results, i.e., the same controller and performance bounds, as if we used an SMDP as per Definition~\ref{def:mdp_st_abstraction}.
\begin{theorem}[SMDP Abstraction with Smaller Partition]
\label{thm:optimality_smdp_smaller_partition}
    Consider System $\Scal$, the $R$-respecting partition $B^\eta$ and the partition $C^\eta$ of $W$. Let $\widetilde  B^\eta := \big( \cup_{b\in B^\eta: b \subseteq \Xsafe} \{b\} \big) \cup\{\Xunsafe\}$, $\widetilde S^\eta := \big( S^\eta \cap \Xsafe\big) \cup \{s_\unsafe\}$, with $s_\unsafe$ being the representative point of the region $\Xunsafe$ Let $\tilde \rho$ be the same as $\rho$ for states $s\in \Xsafe$, $\tilde \rho^{\eta}(x) = s_\unsafe$ for all $x\in \Xunsafe$ and $\tilde\rho^{-\eta}(s_\unsafe) = \Xunsafe$. Define the abstraction function $\widetilde{\mathrm{Abs}}_{\mathbb{S}}$ that abstracts System $\Scal$ into the SMDP $\widetilde \Ucal^{\mathbb{S}, \eta} := (\widetilde S^\eta, A,\widetilde \Gamma^{\mathbb{S}, \eta}, \AP, \widetilde L)$ as per Definition~\ref{def:mdp_st_abstraction}, but using the sets $\widetilde S^\eta$ and $\widetilde B^\eta$ and the function $\tilde \rho$ defined here. Then, 
    \begin{itemize}
        \item $\widetilde{\mathrm{Abs}}_\mathbb{S}$ is asymptotically optimal,
        \item Algorithm~\ref{alg:algo} and $\widetilde{\mathrm{Abs}}_\mathbb{S}$ solve Problem~\ref{prob:prob} in finite time.
    \end{itemize}
\end{theorem}
%

\subsection{IMDPs}
An IMDP $\Ucal^{\mathbb{I}}$ is a UMDP class whose ambiguity set $\Gamma^{\mathbb{S}}$ constrains the probability of transitioning to each state by an interval. In the following definition we describe the IMDP abstraction procedure of \cite{10347405}.
\begin{definition}[IMDP Abstraction \cite{10347405}]
    \label{def:imdp_abstraction}
    Consider System $\Scal$, the $R$-respecting partition $B^\eta$ and the partition $C^\eta$ of $W$. For each $(s, a, s') \in S_\safe^\eta \times A \times S^\eta$, let 
    \begin{align*}
        &\underline P^\eta(s,a,s') \\
        &:=\sum_{c \in C^\eta}P_W(c) \boldsymbol{1}\Big(\mathcal{R}(\rho^{-\eta}(s), a, c) \subseteq \rho^{-\eta}(s') \Big),\\
        &\overline P^\eta(s,a,s') \\
        &:=\sum_{c \in C^\eta}P_W(c) \boldsymbol{1}\Big(\mathcal{R}(\rho^{-\eta}(s), a, c) \cap \rho^{-\eta}(s') \neq \emptyset \Big),
    \end{align*}
    where $\boldsymbol{1}$ returns $1$ if its argument is true and $0$ otherwise.
    We define the \emph{Interval MDP} (IMDP) abstraction of System $\Scal$ as the tuple $\Ucal^{\mathbb{I}, \eta} = (S^\eta, A,\Gamma^{\mathbb{I}, \eta}, \AP, L)$, where 
    \begin{multline}
        \label{eq:Gamma_poly}
        \Gamma_{s,a}^{\mathbb{I}, \eta} := \big\{ \gamma \in \mathcal{P}(S^\eta) : \: \forall s' \in S^\eta,\\
        \underline P^\eta(s,a,s') \le \sum_{s'\in S^\eta} \gamma(s') \le \overline P^\eta(s,a,s')
         \big\}
    \end{multline}
for all $s\in S_\safe^\eta$ and $a\in A$, and $\Gamma_{s, a}^{\mathbb{I}, \eta} := \{\delta_s\}$ for all $s \notin S^\eta_\safe$ for all $a \in A$.
\end{definition}
\cite{10347405} proves that such an IMDP abstraction is Sound as per Definition~\ref{dfn:soundness}. The following theorem shows that this abstraction procedure might fail under very mild assumptions.
\begin{theorem}
\label{thm:optimality_imdp}
    Consider the abstraction procedure $\mathrm{Abs}_\mathbb{I}$ that abstracts $\Scal$ into an IMDP via a uniform partition and let $A = \{a\}$. Assume that $f$ is not one-step contractive, namely, there exist a set $X_0 \subset \Xsafe$ and $w \in W$ such that $\mathrm{diam}(\mathcal{R}(X_0,a,\{w\})) > \mathrm{diam}(X_0)$. Then,
    \begin{itemize}
        \item $\lim_{\eta\to0} \phi(\mathrm{Abs}_\mathbb{I}(\Scal,\eta)) \neq 0$,
        \item $\mathrm{Abs}_\mathbb{I}$ is not asymptotically optimal,
        \item There exists a specification $\varphi$ and $\varepsilon > 0$ for which Algorithm~\ref{alg:algo} does not solve Problem~\ref{prob:prob}.
    \end{itemize}
\end{theorem}
The proof is by counterexample. We show that for specific system dynamics and a simple reach-avoid specification, strategy synthesis on the IMDP abstraction yields trivial bounds, i.e., $\underline p_1(x_0) = 0$ and $\overline p_1(x_0) = 1$ in the probability of satisfying the specification for some state $x_0$. We show that these bounds cannot be tightened by refining the partitions.

We note that, even though it is common practice to employ partitions of different granularities, adapting Theorem~\ref{thm:optimality_imdp} to consider such partitions becomes challenging. However, we highlight that the key phenomenon in the proof of the theorem, i.e., some reachable sets intersecting more than one region in the partition, would still happen in that case, thus providing evidence that the theorem extends to such partitions.

\section{Empirical Validation}
\label{sec:empirical_validation}

We validate our theoretical results through two case studies: a $1$D temperature regulation benchmark with a reach-avoid goal and a $2$D nonlinear cart model with an \LTLf goal. For more details on the case studies, see
Appendix~\ref{app:additional_results}.

The $1$D system is given in \cite{gracia2025data}, and its results are shown in Figure~\ref{fig:comparison}. We compare the bounds in the satisfaction probabilities of SMDPs and IMDPs. The goal is to regulate the system's temperature to the interval $r_\mathrm{goal} =[21.25, 20.75]$ under $T = 20$ steps while staying within the bounds $\Xsafe = [19, 22]$. Figure~\ref{fig:comparison}(b) shows the iterations of Algorithm~\ref{alg:algo}, plotting the termination criterion for IMDP and SMDP abstractions as a function of the iteration number. As expected, while the SMDP-based approach asymptotically converges to $0$, shrinking the error bound within $2\%$ in $8$ iterations, the results of the IMDP approach do not improve via refinement. Figure~\ref{fig:comparison}(a) shows the bounds in the satisfaction probability for the $8$th iteration, with $|S^\eta| = 1280$ and $|C^\eta| = 256$.
\begin{figure}[t]
    \centering
    \begin{subfigure}{0.49\linewidth}
        \centering
        \includegraphics[width=\linewidth]{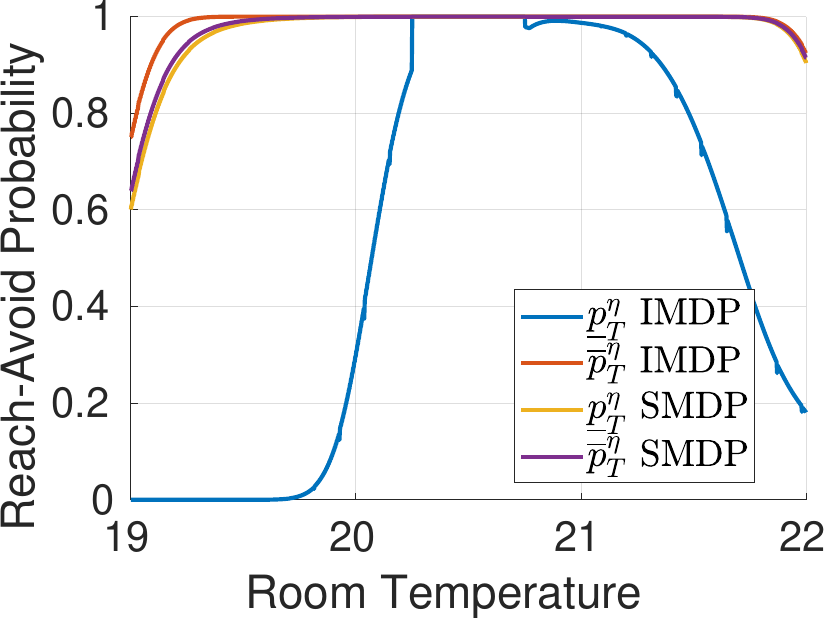}
        \caption{}
    \end{subfigure}
    \hfill
    \begin{subfigure}{0.49\linewidth}
        \centering
        \includegraphics[width=\linewidth]{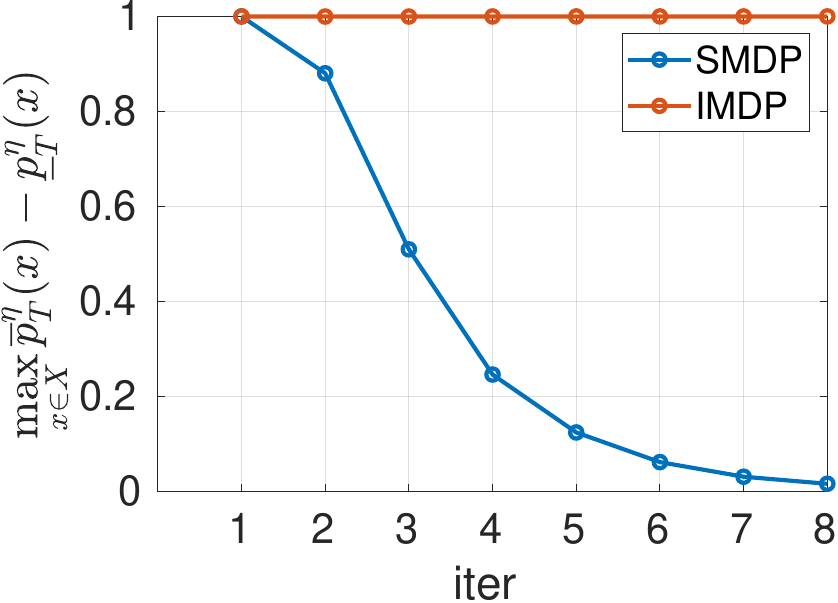}
        \caption{}
    \end{subfigure}
    \caption{Temperature regulation benchmark: (a) Lower and upper bounds in the satisfaction probabilities for fixed discretization. (b) Gap in the satisfaction probability bounds v.s. iteration number in Algorithm~\ref{alg:algo}}
\label{fig:comparison}
\end{figure}
For further illustrations, we also studied a $2$D nonlinear model of a cart under an \LTLf specification $\varphi_2$ taken from \cite{vazquez2018learning}. 
The goal is to reach the $\mathrm{charge}$ station while remaining in $\Xsafe$, which excludes the obstacles, and if the system travels through a region with $\mathrm{water}$, then it must first dry in the $\mathrm{carpet}$ before reaching the $\mathrm{charge}$ station. We seek to satisfy the specification within $T=60$ time steps. Figure~\ref{fig:2d} shows the setup, the satisfaction probability bounds obtained with the SMDP and IMDP approaches, for two partitions, as well as closed-loop trajectories. The trend is similar to the one observed for the $1$D case study. 
\begin{figure}[h]
    \centering

    \begin{subfigure}{0.49\linewidth}
        \centering
        \includegraphics[width=\linewidth]{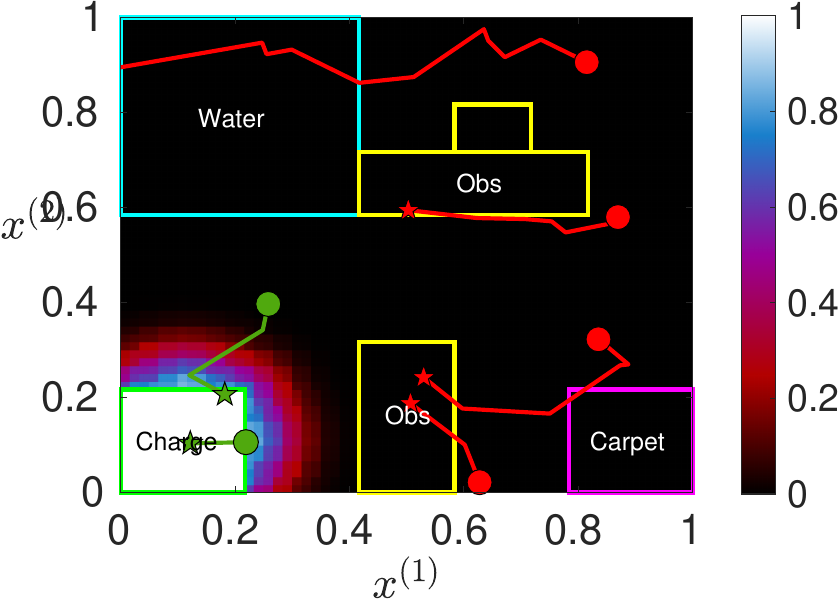}
        \caption{IMDP coarse}
    \end{subfigure}
    \begin{subfigure}{0.49\linewidth}
        \centering
        \includegraphics[width=\linewidth]{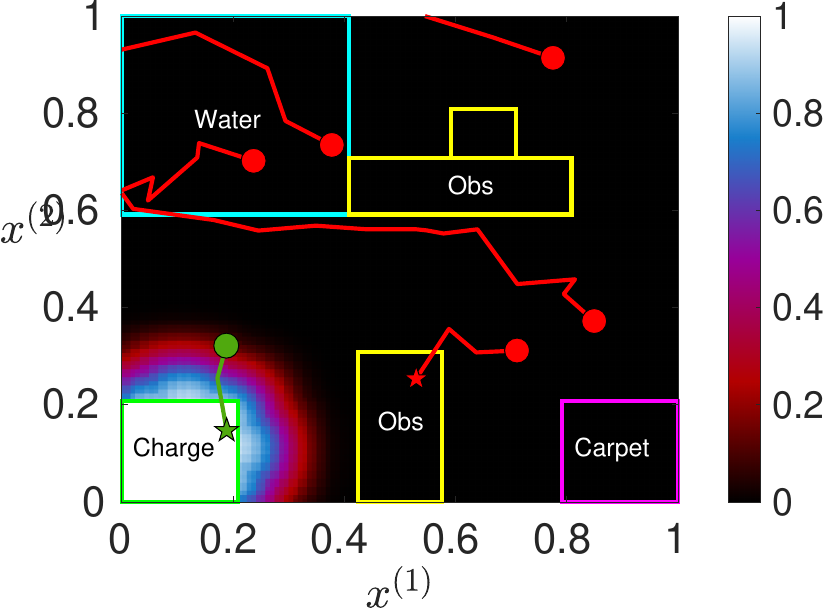}
        \caption{IMDP fine}
    \end{subfigure}

    \begin{subfigure}{0.49\linewidth}
        \centering
        \includegraphics[width=\linewidth]{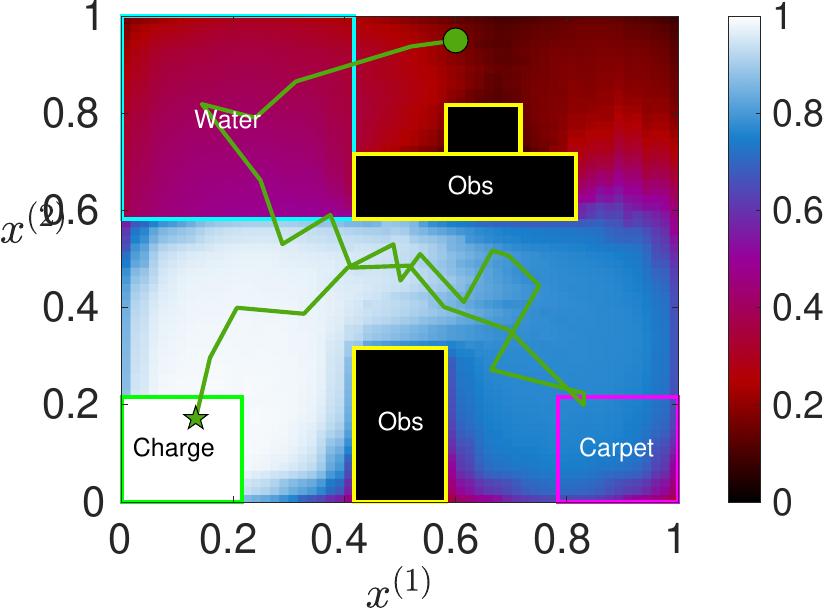}
        \caption{SMDP coarse}
    \end{subfigure}
    \begin{subfigure}{0.49\linewidth}
        \centering
        \includegraphics[width=\linewidth]{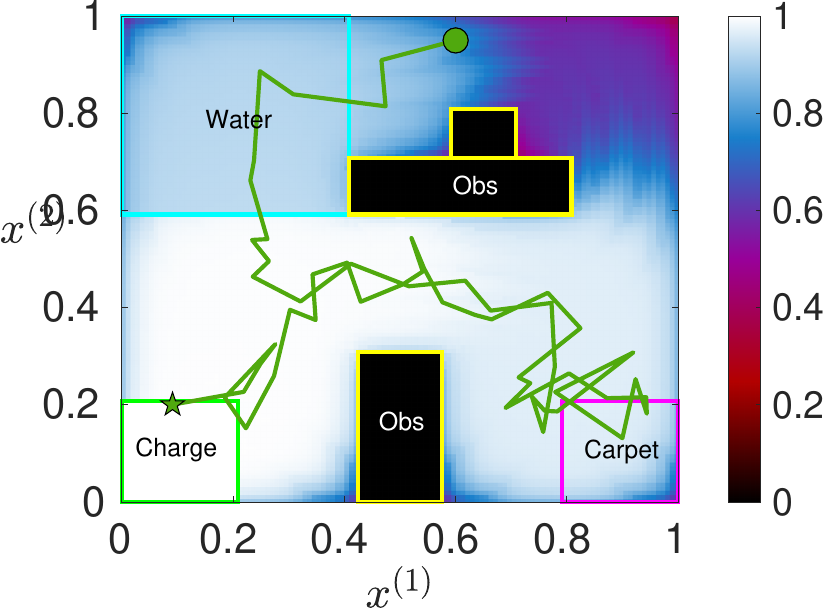}
        \caption{SMDP fine}
    \end{subfigure}
    \caption{$2$D case study. Lower bound in the probability of satisfying $\varphi_2$ within $T = 60$ time steps. The color indicates the value of $\underline p_T^\eta(x_0)$. In green, trajectories that satisfy $\varphi_2$. In red, those that do not.}
    \label{fig:2d}
\end{figure}

\appendix
\subsection{Proofs}
\label{app:proofs}

\begin{proof}[Proof of Lemma~\ref{lemma:uniform_convergence_kernel}]
    Pick $x,y\in X$, $a\in A$ and a coupling $\pi \in \arg\min_{\pi \in \Pi(\Tau(\cdot \mid x,a), \Tau(\cdot \mid y,a))} \int_{X\times X} \|x' - y'\| d\pi(x',y')$. Then, it holds that $|\int_X v d\Tau(\cdot \mid a, x) - \int_X v d\Tau(\cdot \mid a, y)| = |\int_{X\times X} v(x') d\pi(x',y') - \int_{X\times X} v(y') d\pi(x',y')| \le |\int_{X\times X} (v(x') - v(y')) d\pi(x',y')| \le \int_{X\times X} |v(x') - v(y')| d\pi(x',y') \le \int_{X\times X} \mu(\|x'-y'\|) d\pi(x',y')$. By Jensen's inequality, the previous expression is bounded by $\mu( \int_{X\times X} \|x'-y'\| d\pi(x',y') ) = \mu(\W(\Tau(\cdot \mid x,a), \Tau(\cdot \mid y,a)))$. Standard results on Wasserstein distances give \cite{villani2009optimal} $\W(\Tau(\cdot \mid x,a), \Tau(\cdot \mid y,a)) \le L_f \|x-y\|$, which concludes the proof.
\end{proof}

\begin{proof}[Proof of Lemma~\ref{lemma:uc_value_functions}]
    We  proceed by induction. Assume that $v_k(\cdot,z)$ is uniformly continuous on $X$ for each $z$, which is trivially satisfied at $k = 0$. We want to prove uniform continuity of $v_{k+1}$. Let $x,y \in X$ and $z \in Z$, and note that 
    \begin{multline*}
        |v_{k+1}(x,z) - v_{k+1}(y,z)| = |\Bcal[v_k](x,z) - \Bcal[v_k](y,z)| \\
    = |\max_{a\in A} \Bcal_a[v_k](x,z) - \max_{a\in A} \Bcal_a[v_k](y,z)|
    \end{multline*}
    Let $a^* \in\arg\max\{ \Bcal_a[v_k](x,z) : a\in A\}$. Then, 
\begin{multline*}
    |\max_{a\in A} \Bcal_a[v_k](x,z) - \max_{a\in A} \Bcal_a[v_k](x,z)|\\
    \le | B_{a^*}^\varphi[v_k](x,z) - B_{a^*}^\varphi[v_k](x,z)|\\
    \le \max_{a\in A} |\Bcal_a[v_k](x,z) - \Bcal_a[v_k](x,z)|.
\end{multline*}
Next, we bound the right-hand side. Let $D_\epsilon := \big( \cup_{r\in R}\partial r \big) \oplus B(0,\epsilon)$ be the $\epsilon$-thickening of the zero-measure set $\cup_{r\in R}\partial r$. Fix $z\in Z$ and note that $v_z := v_k(\cdot, \delta(z,L(\cdot))) \in \Vcal$ is uniformly continuous on each $r \in R$, because $v_k(\cdot, z')$ is uniformly continuous on $X$ for each $z'\in Z$ and $\delta(z, L(\cdot))$ is constant on each $r\in R$. This result, in conjunction with finiteness of $|R|$, makes $v_z$ uniformly continuous on $X\setminus D_\epsilon$, with some modulus of continuity $\mu_{v_z, \epsilon}$ that is concave and non-decreasing. By \cite{mcshane1934extension}, there exists a function $v_\epsilon' :X\rightarrow \reals$ that is uniformly continuous on $X$, agrees with $v_z$ on $X\setminus D_\epsilon$ and for which $\mu_{v_z,\epsilon}$ is also a valid modulus of continuity. We define $v_\epsilon\in\Vcal$ as $v_\epsilon(x) := \max\{0,\min\{v_\epsilon'(x), 1\}\}$ for all $x\in X$, and note that $\mu_{v_z, \epsilon}$ is also a modulus of continuity of $v_\epsilon$. Define also $r_\epsilon \in \Vcal$ as $r_\epsilon := v_z - v_\epsilon$, which is zero on $X \setminus D_\epsilon$. Pick $a\in A$. Since $|r_\epsilon| \le 1$, then
    \begin{align*}
        &|\Bcal_a[v_k](x,z) - \Bcal_a[v_k](y,z)| \le\\
        &\le | \int_{X} v_z  d\Tau(\cdot \mid x, a) -\int_{X} v_z d\Tau(\cdot \mid y, a) |\\
        &\le | \int_{X} v_\epsilon d\Tau(\cdot \mid x, a) - \int_{X} v_\epsilon d\Tau(\cdot \mid y, a) | +\\
        &+ | \int_{X} r_\epsilon d\Tau(\cdot \mid x, a) - \int_{X} r_\epsilon d\Tau(\cdot \mid y, a) |\\
        &\le | \int_{X} v_\epsilon d\Tau(\cdot \mid x, a) - \int_{X} v_\epsilon d\Tau(\cdot \mid y, a) | +\\
        &+\Tau(D_\epsilon \mid x, a) + \Tau(D_\epsilon \mid y, a).
    \end{align*}
    By Proposition~\ref{lemma:uniform_convergence_kernel}, 
    $| \int_{X} v_\epsilon d\Tau(\cdot \mid x, a) - \int_{X} v_\epsilon d\Tau(\cdot \mid y, a) | \le \mu_{v_z,\epsilon}(L_f\|x-y\|)$. 
    %
    %
    Taking also into account Assumption~\ref{ass:density}, the bound becomes
    \begin{align*}
        | \int_{X} v_z d\Tau(\cdot \mid x, a) - \int_{X} v_z d\Tau(\cdot \mid y, a) |\\ \le \mu_{v_z,\epsilon}(L_f\|x-y\|) + 2\overline \tau \mathcal{L}(D_\epsilon).
    \end{align*}
    Let $\theta > 0$ and choose $\epsilon$ such that $2\overline \tau \mathcal{L}(D_\epsilon) \le \theta/2$. Given that $\epsilon$, pick $\beta>0$ such that $\mu_{v_z,\epsilon} (L_f\|x-y\|) \le \theta/2$ for all $\|x-y\| \le \beta$, which is always possible since $\mu_{v_z,\epsilon}$ is a valid modulus of continuity. It is evident that for all $x,y\in X$ such that $\|x - y\| \le\beta$ it holds that $| \int_{X} v_z d\Tau(\cdot \mid x, a) - \int_{X} v_z d\Tau(\cdot \mid y, a) | \le \theta$. Since such a $\beta$ exists for each $\theta > 0$ and $a\in A$, we conclude that $v_{k+1}(\cdot, z)$ is uniformly continuous on $X$. 
    Furthermore, this holds for each $z\in Z$.
    The proof is concluded by induction.
\end{proof}

\begin{proof}[Proof of Lemma~\ref{lemma:aux_2_term}]
    Similarly to the proof of Lemma~\ref{lemma:uc_value_functions}, let $D_\epsilon := \cup_{r\in R}\partial r \oplus B(0, \epsilon)$. Fix $z\in Z$ and $a\in A$. Note that $v_z$ is uniformly continuous on each $r \in R$, because $v(\cdot, z')$ is uniformly continuous on $X$ for each $z'\in Z$ and $\delta(z, L(\cdot))$ is constant on each $r\in R$. This result, in conjunction with finiteness of $|R|$, makes $v_z$ uniformly continuous on $X\setminus D_\epsilon$, with some modulus of continuity $\mu_{v_z, \epsilon}$ that is concave and non-decreasing. By \cite{mcshane1934extension}, there exists a function $v_\epsilon' :X\rightarrow \reals$ that is uniformly continuous on $X$, agrees with $v_z$ on $X\setminus D_\epsilon$ and for which $\mu_{v_z,\epsilon}$ is also a valid modulus of continuity. We define $v_\epsilon\in\Vcal$ as $v_\epsilon(x) := \max\{0,\min\{v_\epsilon'(x), 1\}\}$ for all $x\in X$, and note that $\mu_{v_z, \epsilon}$ is also a modulus of continuity for $v_\epsilon$. Define also $r_\epsilon \in \Vcal$ as $r_\epsilon = v_z - v_\epsilon$, which is zero on $X \setminus D_\epsilon$. By Proposition~\ref{lemma:uniform_convergence_kernel} and since $|r_\epsilon| \le 1$, 
    \begin{align*}
        &| \int_X v_z d\gamma_{x,a}^\eta - \int_X v_z d\Tau^\eta(\cdot \mid x, a) |\\
        &\le | \int_X v_\epsilon d\gamma_{x,a}^\eta - \int_X v_\epsilon d\Tau^\eta(\cdot \mid x, a) | +\\
        &+ | \int_X r_\epsilon d\gamma_{x,a}^\eta - \int_X r_\epsilon d\Tau^\eta(\cdot \mid x, a) |\\
        &\le \mu_{v_z,\epsilon}(\W(\gamma_{x,a}^\eta, \Tau^\eta(\cdot \mid x, a))) + \\
        &+ \gamma_{x,a}^\eta(D_\epsilon) + \Tau^\eta(D_\epsilon \mid x, a).
    \end{align*}
    By soundness of $\Ucal^\eta$, $\Tau^\eta(\cdot \mid x,a) \in \Gamma_{s,a}^\eta$. Since $\W(\gamma_{x,a}^\eta, \Tau^\eta(\cdot \mid x, a)) \le \mathrm{diam}_\W(\Gamma_{s,a}^\eta) \le \phi(\eta)$. Next, we bound $\gamma_{x,a}^\eta(D_\epsilon)$. The previous result implies that there exists a coupling $\pi \in \Pcal(X\times X)$ with marginals $\Tau^\eta(\cdot \mid x, a)$ and $\gamma_{x,a}^\eta$, implying that we can express $\gamma_{x,a}^\eta(D_\epsilon) = \pi(X \times D_\epsilon)$, and which has transport cost $\int_{X\times X} |x' - y'|d\pi(x',y') \le \phi(\eta)$. Now, the event $X \times D_\epsilon$ satisfies
    \begin{align*}
        \Xabs\times D_\epsilon &= \{(x,y)\in X\times X : y \in D_\epsilon\}\\
        &\subseteq \{(x',y')\in X\times X : x' \in D_{2\epsilon}\} \cup\\
        &\cup \{(x',y')\in X\times X : |x'-y'| > \epsilon\}. 
    \end{align*}
    Therefore, 
    \begin{align*}
        \gamma_{x,a}^\eta(D_\epsilon) \le \Tau^\eta(D_{2\epsilon} \mid x, a) + \pi(\{(x',y') : |x'-y'| > \epsilon\}),
    \end{align*}
    where Chebyshev's inequality gives us $\pi(\{(x',y') : |x'-y'| > \epsilon\}) \le (1/\epsilon)\int_{X\times X}|x' - y'|d\pi(x',y') \le \phi(\eta)/\epsilon$. Next, we bound $\Tau^\eta(D_{2\epsilon} \mid x, a)$. Noting that $\Tau^\eta(\cdot \mid x, a)$ and $\Tau(\cdot \mid x, a)$ are supported on $\Xabs$ and $\mathrm{diam}(b) \le 2\eta$ for all regions $b\in B^\eta$,
    \begin{align*}
        &\sum_{b\in B^\eta} \Tau^\eta(b\cap D_{2\epsilon}  \mid x, a) = \sum_{b\in B^\eta : b\cap D_{2\epsilon} \neq \emptyset} \Tau^\eta(b \cap D_{2\epsilon} \mid x, a)\\
        &\le \sum_{b\in B^\eta : b\cap D_{2\epsilon} \neq \emptyset} \Tau^\eta(b \mid x, a) = \Tau(\bigcup_{b\in B^\eta
        : b\cap D_{2\epsilon} \neq \emptyset} b \mid x, a)
        \\
        &\le \Tau(D_{2\epsilon + 2\eta}  \mid x, a) \le \overline \tau \mathcal{L}(D_{2\epsilon + 2\eta}).
    \end{align*}
    %
    %
    Putting the previous bounds together we obtain that, for $\epsilon>0$ small enough,
    \begin{align*}
        &| \int_X v_z d\gamma_{x,a}^\eta - \int_X v_z d\Tau^\eta(\cdot \mid x, a) |\\
        &\le \mu_{v_z,\epsilon}(\phi(\eta)) + 
        \phi(\eta)/\epsilon + \overline \tau (\mathcal{L}(D_{2\epsilon + 2\eta}) + \mathcal{L}(D_{\epsilon + 2\eta}) ).
    \end{align*}
    Let $\theta > 0$ and choose $\epsilon, \eta_0 > 0$ such that $\overline \tau (\mathcal{L}(D_{2\epsilon + 2\eta}) + \mathcal{L}(D_{\epsilon + 2\eta}) ) \le \theta/2$ for all $\eta \in (0, \eta_0]$. Given that $\epsilon$, pick $\eta_0' \in (0, \eta_0]$ such that $\mu_{v_z,\epsilon}(\phi(\eta)) + \phi(\eta)/\epsilon \le \theta/2$ for all $\eta \in (0, \eta_0']$, which is always possible since $\mu_{v_z,\epsilon}$ is a valid modulus of continuity. Note that picking $\eta_0' \le \eta_0 $ only makes the terms $\overline \tau (\mathcal{L}(D_{2\epsilon + 2\eta}) + \mathcal{L}(D_{\epsilon + 2\eta}) )$ smaller. It is evident that for all $\eta\in (0, \eta_0']$ it holds that $| \int_X v_z d\gamma_{x,a}^\eta - \int_X v_z d\Tau^\eta(\cdot \mid x, a) | \le \theta$. Since such a $\eta_0'$ exists for each $\theta > 0$ independently from $x\in \Xsafe$, we conclude that $\lim_{\eta\to 0}| \int_X v_z d\gamma_{x,a}^\eta - \int_X v_z d\Tau^\eta(\cdot \mid x, a) |=0$  uniformly in $x \in \Xsafe$. The proof is concluding by observing that $z\in Z$ and $a\in A$ were arbitrary choices.
\end{proof}

\begin{proof}[Proof of Lemma~\ref{lemma:aux_3_term}]
For each $r\in R$, let $\mu_{v,z,r}$ denote a modulus of continuity of $v_z$ over $r$, meaning that
\begin{align*}
    | v_z(x') - v_z(y') | \le \mu_{v,z,r}( \|x' - y'\| ),
\end{align*}
for all $x', y' \in r$. Let also $\mu_{v,z}(\eta) := \max_{r\in R} \mu_{v,z,r}(\eta)$ for all $\eta \ge 0$. Denote by $r(s') \in R$ the region of interest that contains $s'\in S^\eta$. Then, since $\Tau^\eta(\cdot \mid x, a)$ and $\Tau(\cdot \mid x, a)$ are supported on $\Xabs$,
\begin{multline*}
    |\int_X v_z d\Tau^\eta(\cdot \mid x, a) - \int_X v_z d\Tau(\cdot \mid x, a) |\\
    \le \sum_{b\in B^\eta} \bigg|  \Tau^\eta(b \mid x, a) \sup_{x'\in b} v_z(x') - \Tau(b \mid x, a) \inf_{y'\in b} v_z(y') \bigg|\\
    = \sum_{b \in B^\eta} \Tau(b \mid x, a) \sup_{x', y'\in b} | v_z(x')  -  v_z(y') |\\
    \le \sum_{s'\in S^\eta} \Tau(\rho^{-\eta}(s') \mid x, a) \mu_{v,z,r(s')}(2\eta)
    \le \mu_{v_z}(2\eta),
\end{multline*}
%
uniformly on $x\in \Xsafe$. Observing that $\lim_{\eta \to 0} \mu_v(2\eta) = 0$ concludes the proof.
\end{proof}

\begin{proof}[Proof of Lemma~\ref{lemma:uc_operator}]
Let $x \in \Xsafe$, $s = \rho(x)$ and $z \in Z$. Note that
\begin{multline*}
    |\underline \Bcal^\eta[v](x,z) - \Bcal[v](x,z)|\\
    = |\max_{a\in A} \underline \Bcal_a^{\eta}[v](x,z) - \max_{a\in A} \Bcal_a[v](x,z)|.
\end{multline*}
Let $a^* \in\arg\max\{\underline \Bcal_a^{\eta}[v](x,z) : a\in A\}$. Then, 
\begin{multline}
\label{eq:max_a}
    |\max_{a\in A} \underline \Bcal_a^{\eta}[v](x,z) - \max_{a\in A} \Bcal_a[v](x,z)|\\
    \le |\underline \Bcal_{a^*}^{\eta, \varphi}[v](x,z) - B_{a^*}^\varphi[v](x,z)|\\
    \le \max_{a\in A} |\underline \Bcal_a^{\eta}[v](x,z) - \Bcal_a[v](x,z)|.
\end{multline}
Therefore, it suffices to prove uniform convergence of $|\underline \Bcal_a^{\eta}[v](x,z) - \Bcal_a[v](x,z)|$ for all $a\in A$. Pick $a \in A$ and denote $$\gamma^{\eta*}_{x,a} \in \arg\min_{\gamma\in\Gamma_{s,a}^\eta}\int_X v^\eta(\cdot,\delta(z,L(\cdot))) d\gamma.$$ Let $v_z := v(\cdot,\delta(z,L(\cdot)))$ and $v_z^\eta := v^\eta(\cdot,\delta(z,L(\cdot)))$. Then,
\begin{multline}
\label{eq:uc_terms}
    |\underline \Bcal_a^{\eta}[v](x,z) - \Bcal_a[v](x,z)| \\
    \le | \int_X v_z^\eta d\gamma_{x,a}^{\eta*} - \int_X v_z d\Tau(\cdot \mid x, a) |\\
    \le | \int_X v_z^\eta d\gamma_{x,a}^{\eta*} - \int_X v_z d\gamma_{x,a}^{\eta*} |\\
    + | \int_X v_z d\gamma_{x,a}^{\eta*} - \int_X v_z d\Tau^\eta(\cdot \mid x, a) |\\
    + | \int_X v_z d\Tau^\eta(\cdot \mid x, a) - \int_X v_z d\Tau(\cdot \mid x, a) |.
\end{multline}
Noting that $\gamma_{x,a}^{\eta*}$ and $\Tau^\eta(\cdot \mid x, a)$ are supported on $\Xabs$, we obtain the following bound on the first term in Equation~\eqref{eq:uc_terms}: 
\begin{align*}
    & \sum_{b\in B^\eta} \int_b |v_z^\eta(x') - v_z(x')| d\gamma_{x,a}^{\eta*}(x')\\
    \le& \sum_{b\in B^\eta} | \sup_{x'\in b} v_z(x') - \inf_{x'\in b} v_z(x') | \gamma_{x,a}^{\eta*}(b)\\
    \le& \sum_{b\in B^\eta} \sup_{x', y'\in b} |  v_z(x') - v_z(x') | \gamma_{x,a}^{\eta*}(b)\\
    \le& \sum_{s'\in S^\eta} \mu_{v,z,r(s')}(2\eta) \gamma_{x,a}^{\eta*}(\rho^{-\eta}(s')) \le \mu_{v}(2\eta).
\end{align*}
%
Since the previous bound holds for all $x\in \Xsafe$, the first term in \eqref{eq:uc_terms} converges to $0$ uniformly on $x\in \Xsafe$.

Next, we bound the second term in \eqref{eq:uc_terms}. Since $\phi(\eta)$ and the function $v_z \in \Vcal$ satisfy the conditions of Lemma~\ref{lemma:aux_2_term} and $\gamma_{x,a}^{\eta,*} \in \Gamma_{s,a}^\eta$, the second term in \eqref{eq:uc_terms} converges to $0$ uniformly on $x\in \Xsafe$.

Finally we bound the third term in \eqref{eq:uc_terms}. Since $v$ satisfies the conditions of Lemma~\ref{lemma:aux_3_term}, the third term in \eqref{eq:uc_terms} converges to $0$ uniformly over $x\in \Xsafe$.

By the previous reasoning, $|\underline \Bcal_a^{\eta}[v](x,z) - \Bcal_a[v](x,z)|$ converges to $0$, as $\eta \to 0$, uniformly on $x \in \Xsafe$ for all actions $a\in A$. By finiteness of $A$, it follows that $|\underline \Bcal^\eta[v](x,z) - \Bcal[v](x,z)|$ also converges to $0$ uniformly on $x \in \Xsafe$, concluding the proof.
\end{proof}
Before proving Theorem~\ref{thm:asymptotic_optimality} we state the following technical lemma.
\begin{lemma}
\label{lemma:non_expansive}
    Let $v,v'\in \Vcal^\varphi$ with $v(\cdot, z_\text{unsafe}) = v'(\cdot, z_\text{unsafe}) = 0$. Then,
    \begin{multline*}
        \sup_{x\in \Xsafe, z\in Z} |\underline \Bcal^\eta[v](x,z) - \underline \Bcal^\eta[v'](x,z)|\\
        \le \sup_{x\in \Xsafe, z\in Z} |v^\eta(x, z) - v'^\eta(x,z)|.
    \end{multline*}
\end{lemma}
\begin{proof}
    Pick $v, v' \in \Vcal^\varphi$, $x\in \Xsafe$ and let $s = \rho(x)$. Define $v_z := v(\cdot, \delta(z,L(\cdot)))$ and $v_z' := v'(\cdot, \delta(z,L(\cdot)))$ for all $z\in Z$. Then, by \eqref{eq:max_a}, $|\underline \Bcal^\eta[v](x,z) - \underline \Bcal^\eta[v'](x,z)| \le \max_{a\in A}|\underline \Bcal_a^\eta[v](x,z) - \underline \Bcal_a^\eta[v'](x,z)|$. Fix $a\in A$ and define $\gamma'^* \in \arg\min_{\gamma\in\Gamma_{s,a}^\eta}\int_X v_z'^\eta d\gamma$. Then, noting that $\gamma'^*$ is supported on $\Xabs$, we obtain
\begin{align*}
    |\underline \Bcal_a^\eta[v](x,z) - \underline \Bcal_a^\eta[v'](x,z)|\\
    \le |\int_X v_z^\eta d\gamma'^* - \int_X v_z'^\eta d\gamma'^*| \le \sum_{b \in B^\eta} \int_b |v_z^\eta - v_z'^\eta| d\gamma'^*\\
    \le \sum_{b \in B^\eta} \gamma'^*(b) \sup_{x'\in b}|v^\eta(x', \delta(z,L(x'))) -\\
    - v'^\eta(x', \delta(z,L(x')))|.
\end{align*}
%
Taking into account that $\delta(z, L(x')) = z_\text{unsafe}$ for all $x' \notin \Xsafe$ and that $v(\cdot, z_\text{unsafe}) = v'(\cdot, z_\text{unsafe}) = 0$, the previous expression is bounded by
\begin{align*}
    \sum_{b \in B^\eta} \gamma'^*(b) \sup_{x'\in \Xsafe}|v^\eta(x', \delta(z,L(x'))) -\\
    - v'^\eta(x', \delta(z,L(x')))| \le \sup_{x'\in \Xsafe, z\in Z} |v^\eta(x', z) - v'^\eta(x',z)|.
\end{align*}
Since the previous bound is uniform on $a\in A$, it is also a bound on $|\underline \Bcal^\eta[v](x,z) - \underline \Bcal^\eta[v'](x,z)|$. Observing that $x\in \Xsafe$ and $z\in Z$ were chosen arbitrarily concludes the proof.
\end{proof}

\begin{proof}[Proof of Theorem~\ref{thm:asymptotic_optimality}]
Consider the RDP sequences $(v_k)_{k \le T}, (\underline v_k^\eta)_{k \le T} \subset\Vcal^\varphi$ defined in Proposition~\ref{prop:RDP}. Since $z_\text{unsafe}$ is absorbing, one can prove by induction that $v_k(\cdot, z_\text{unsafe}) = \underline v_k^\eta(\cdot, z_\text{unsafe}) = 0$ for all $k\le T$. By Lemma~\ref{lemma:non_expansive},
%
%
\begin{multline}
    \label{eq:DP_aux}
    \sup_{x\in \Xsafe, z\in Z} |v_T(x,z) - \underline v_T^\eta(x,z)| \\
    = \sup_{x\in \Xsafe, z\in Z} |\Bcal[v_{T-1}](x,z) - \underline \Bcal^\eta[\underline v_{T-1}^\eta](x,z)| \\
    \le \sup_{x\in \Xsafe, z\in Z}|\Bcal[v_{T-1}](x,z) - \underline \Bcal^\eta[v_{T-1}](x,z)| + \\
    + \sup_{x\in \Xsafe, z\in Z}|\underline \Bcal^\eta[v_{T-1}](x,z) - \underline \Bcal^\eta[\underline v_{T-1}^\eta](x,z)|\\
    \le \sup_{x\in \Xsafe, z\in Z}|\Bcal[v_{T-1}](x,z) - \underline \Bcal^\eta[ v_{T-1}](x,z)| + \\
    + \sup_{x\in \Xsafe, z\in Z}|v_{T-1}(x,z) - \underline v_{T-1}^\eta(x,z)|\\
    \le \sum_{k = 0}^{T-1} \sup_{x\in \Xsafe, z\in Z}|\Bcal[v_k](x,z) - \underline \Bcal^\eta[v_k](x,z)|.
\end{multline}

Pick $\epsilon > 0$, let $k \in \{0, 1, \dots, T-1\}$ and note that, by Lemma~\ref{lemma:uc_operator}, for all $\epsilon>0$ I can pick $\eta_k^0 > 0$ small enough so that for all $\eta \in (0, \eta_k^0 ]$ it holds that $\sup_{x\in \Xsafe, z\in Z}|\Bcal[v_k](x,z) - \underline \Bcal^\eta[v_k](x,z)| < \epsilon/(T-1)$. Then, by Equation~\eqref{eq:DP_aux}, letting $\eta_0 := \min\{\eta_0^k : k \in \{0, 1, \dots, T-1\}\}$ guarantees that $\sup_{x\in \Xsafe, z\in Z}|v_T(x,z) - \underline v_T^\eta(x,z)| < \epsilon$ for all $\eta \in (0, \eta_0]$. Therefore, $\lim_{\eta \to 0} |v_T(x,z) - \underline v_T^\eta(x,z)| = 0$ uniformly in $x\in \Xsafe$ and $z\in Z$. 
This implies that $\underline p_T^\eta(x) = \underline v_T^\eta(x, \delta(z_\mathrm{init}, L(x)))$ converges to $p_T(x) = v_T(x, \delta(z_\mathrm{init}, L(x)))$ uniformly on $x\in \Xsafe$. Furthermore, note that if $x\notin \Xsafe$ it holds that $\delta(z_\mathrm{init},L(x)) = z_\text{unsafe}$, which implies that $p_T(x) := v_T(x,\delta(z_\mathrm{init}, L(x))) = 0$ and $\underline p_T^\eta(x) := \underline  v_T^\eta(x,\delta(z_\mathrm{init}, L(x))) = 0$. This, together with the previous result yields $\lim_{\eta\to 0} |p_T(x) - \underline p_T^\eta(x)| = 0$ uniformly on $X$. 

By a similar reasoning, we observe that the upper bound $\overline p_T^\eta \in \Vcal^\varphi$ defined in Proposition~\ref{prop:RDP} also converges uniformly to $p_T$. Finally, Lemma~\ref{lemma:uc_value_functions} gives us $P_{x_0}^{\kappa^\varepsilon}[\varphi, T] = [\underline p^\varepsilon(x_0), \overline p^\varepsilon(x_0)]$ for all $x_0\in X$. Furthermore, soundness of the abstraction procedure ensures that $P_{x_0}^*[\varphi, T]$ also meets this bounds uniformly on $x_0\in X$. By a similar reasoning as that of \cite[Section 5.1]{jackson2021strategy}, $|P_{x_0}^*[\varphi, T] - P_{x_0}^{\kappa^\varepsilon}[\varphi, T]|\le \overline p(x_0) - \underline p(x_0)|$ which converges to $0$ uniformly on $x_0\in X$. This concludes the proof.

\end{proof}

\begin{proof}[Proof of Corollary~\ref{cor:near_optimal}]
    The proof uses a similar reason than that of \cite[Section 5.1]{jackson2021strategy}. Note that Algorithm~\ref{alg:algo} iteratively refines the partition in such a way that $\eta_i \to 0$. By Theorem~\ref{thm:asymptotic_optimality}, the difference $\overline p(x_0) - \underline p(x_0)$ will eventually be bounded by $\varepsilon$ uniformly on $x_0\in X$ for some iteration number $i\in\naturals_0$, thus exiting the for loop. Next, observe that Proposition~\ref{prop:RDP} gives us $P_{x_0}^{*}[\varphi, T], P_{x_0}^{\kappa^\varepsilon}[\varphi, T] = [\underline p^\varepsilon(x_0), \overline p^\varepsilon(x_0)]$ for all $x_0\in X$. This implies that $|P_{x_0}^*[\varphi, T] - P_{x_0}^{\kappa^\varepsilon}[\varphi, T]|\le \overline p(x_0) - \underline p(x_0)| \le \varepsilon$ for all $x_0\in X$, which concludes the proof.
\end{proof}

\begin{proof}[Proof of Theorem~\ref{thm:optimality_smdp}]
Note that, by Theorem~\ref{thm:asymptotic_optimality} and Corollary~\ref{cor:near_optimal}, the first condition implies the others, thus it suffices to prove the former condition. Pick $x\in X_\safe$ and $a\in A$, and let $s = \rho(s)$. Let also $\gamma, \gamma' \in \Gamma_{s,a}^\eta$ which, by Definition~\ref{def:mdp_st_abstraction}, corresponds
%
%
to the distributions $\alpha_{c_1}, \dots, \alpha_{c_{|C^\eta|}}$, supported on $q_{s,a,c_{1}}, \dots, q_{s,a,c_{|C^\eta|}}$ respectively. Now, for each $c\in C^\eta$, let $\pi_c^* \in \Pcal(q_{s,a,c}^\eta\times q_{s,a,c}^\eta)$ be the optimal coupling between $\alpha_c$ and $\alpha_c'$. Define the coupling $\pi := \sum_{c\in C^\eta} P_W(c) \pi_c^* \in \Pcal(q_{s,a,c}^\eta\times q_{s,a,c}^\eta)$, and note its marginals are
\begin{align*}
    &\sum_{s_1 \in S^\eta} \pi(s_1, s_2) = \sum_{c\in C^\eta}P_W(c) \sum_{s_1\in S^\eta} \pi_c^*(s_1, s_2) =\\
    &= \sum_{c\in C^\eta} P_W(c) \alpha_c'(s_2) = \gamma'(s_2)\\
    &\sum_{s_2\in S^\eta}\pi(s_1, s_2) = \sum_{c\in C^\eta}P_W(c)\sum_{s_2\in S^\eta}\pi_c^*(s_1, s_2) =\\
    &= \sum_{c\in C^\eta}P_W(c)\alpha_c(s_1) = \gamma(s_1),
\end{align*}
for all $s,s' \in S^\eta$, which implies that $\pi$ is a coupling between $\gamma$ and $\gamma'$. Then, the Wasserstein distance between these two distributions is bounded by
\begin{align*}
    &\W(\gamma, \gamma') = \int_{x_1, x_2 \in X} |x - x'|d\pi(x_1, x_2) \\
    &\le \sum_{s_1, s_2\in S^\eta} |s_1 - s_2|\sum_{c\in C^\eta}P_W(c)\pi_c^*(s_1,s_2)\\
    &\le \sum_{c\in C^\eta}P_W(c) \sum_{s_1, s_2\in S^\eta} |s_1 - s_2| \pi_c^*(s_1, s_2)\\
    &= \sum_{c\in C^\eta}P_W(c)\sum_{s_1, s_2 \in q_{s,a,c}^\eta} |s_1 - s_2| \pi_c^*(s_1, s_2)\\ &\le \max_{s_1, s_2 \in q_{s,a,c}^\eta, c\in C^\eta} \|s_1 - s_2\|\\
    &\le \max_{c\in C^\eta} \mathrm{diam}(\mathcal{R}(\rho^{-\eta}(s), a, c)) + 2\eta\\
    &\le (4L_f + 2)\eta.
\end{align*}
Note that this bound holds for all $x\in \Xsafe$ and $a\in A$, and vanishes as $\eta \to 0$. This concludes the proof.
\end{proof}

\begin{proof}[Proof of Theorem~\ref{thm:optimality_smdp_smaller_partition}]
    For the sake of notation, we omit the superscript $\mathbb{S}$ and keep the dependence on $\eta$ implicit. Let $(\underline v_k)_{k \le T}$ and $(\overline v_k)_{k \le T}$ be the sequences of value functions obtained via RDP on $\Ucal^\varphi$ as per Proposition~\ref{prop:RDP}. Similarly, let $(\underline v_k')_{k \le T}$ and $(\overline v_k')_{k \le T}$ be the ones obtained for $\widetilde \Ucal^\varphi$. Note that, by Proposition~\ref{prop:RDP} and an induction argument, one can prove that $\underline v_k(\cdot, z_\unsafe) = \overline v_k(\cdot, z_\unsafe) = \underline v_k'(\cdot, z_\unsafe) = \overline v_k'(\cdot, z_\unsafe) = 0$ for all $k \le T$. We first show, by induction on the sequences $(\underline v_k')_{k\le T}$ and $(\underline v_k)_{k\le T}$, that $\underline p_T'(x) := \underline v_T(x, \delta(z_\mathrm{init}, L(x))) = \underline p_T(x)$ for all $x\in X$. Fix $k \le T$ and assume that $\underline v_k'(x,z) = \underline v_k(x,z)$ for all $x\in \Xsafe$ and $z\in Z$, which is trivially satisfied at $k = 0$. Pick $x\in \Xsafe$ and $z\in Z$, and let $s = \tilde\rho(x)$. Note that, by Proposition~\ref{prop:RDP} and, since $\underline v_k$ is its own piecewise constant minorizer on the partition $B$,
    \begin{align*}
        \underline v_{k+1}(x, z) = \underline \Bcal^\eta[\underline v_k](x,z)\\
        = \mathds{1}_{Z_\text{acc}}(z) + \mathds{1}_{Z\setminus Z_\text{acc}}(z) \max_{a\in A} \min_{\gamma \in \Gamma_{s,a}} \int_X \underline v_k(\cdot, \delta(z, L(\cdot))) d\gamma.
    \end{align*}
    We also obtain a similar expression for $\underline v_{k+1}'$ as a function of $\underline v_k'$. Note that if $z \in Z_\text{acc}$, then $\underline v_{k+1}'(x, z) = \underline v_{k+1}(x, z)$ = 1. We now consider the case $z \notin Z_\text{acc}$. By \cite[Theorem 5]{gracia2025beyond},
    \begin{multline}
    \label{eq:aux_reduced_complexity}
        \underline v_{k+1}(x, z) = \max_{a \in A} \min_{\gamma \in \Gamma_{s,a}} \int_X \underline v_k(\cdot, \delta(z, L(\cdot))) d\gamma\\ 
        = \max_{a\in A} \sum_{c\in C^\eta} P_W(c) \min_{s' \in q_{s,a,c}} \underline v_k(s', \delta(z, L(s'))).
    \end{multline}
    By the previous reasoning we also arrive to a similar expression for the case of $\underline v_{k+1}'$ as a function of $\underline v_k'$. Note that if a cluster $q_{s,a,c}$ of $\Ucal$ contains any unsafe state $s' \in \Xunsafe$, then, the corresponding cluster $\tilde q_{s,a,c}$ of $\widetilde \Ucal$ contains $s_\unsafe$. In that case, since $\delta(z, L(s_\unsafe)) = z_\unsafe$, it follows that $\min_{s' \in q_{s,a,c}} \underline v_k(s', \delta(z, L(s'))) = \min_{s' \in \tilde q_{s,a,c}} \underline v_k'(s', \delta(z, L(s'))) = \underline v_k'(s', z_\unsafe) = 0$. On the other hand, if $q_{s,a,c}$ does not contain any unsafe state, then $\tilde q_{s,a,c} = q_{s,a,c}$. These results, together with Equation~\eqref{eq:aux_reduced_complexity} and the induction hypothesis, imply that $\underline v_{k+1}'(x,z) = \underline v_{k+1}(x,z)$ for all $x\in \Xsafe$ an $z\in Z$. Therefore, by induction, $\underline v_k'(x, z) = \underline v_k(x, z)$ for all $x\in \Xsafe$, $z\in Z$ and $k \le T$. This implies that $\underline p_T'(x) = \underline v_T'(x, \delta(z_\mathrm{init}, L(x))) = \underline v_T(x, \delta(z_\mathrm{init}, L(x))) = \underline v_k(x)$ for all $x\in \Xsafe$. Additionally, for all $x\notin \Xsafe$, $\delta(z_\mathrm{init}, L(x)) = z_\unsafe$, which implies that $\underline p_T'(x) = \underline v_T'(x, \delta(z_\mathrm{init}, L(x))) = \underline v_T(x, \delta(z_\mathrm{init}, L(x))) = \underline v_T(x, z_\unsafe) = 0$. Putting the previous results together we have that $\underline p_T' = \underline p$.
    A similar reasoning reveals that $\overline p_T' = \overline p_T$ Since these results hold for all partition sizes $\eta>0$, asymptotic optimality directly follows from the reasoning in the proof of Theorem~\ref{thm:asymptotic_optimality}, and completeness of Algorithm~\ref{alg:algo} follows from Corollary~\ref{cor:near_optimal}.
\end{proof}
\begin{proof}[Proof of Theorem~\ref{thm:optimality_imdp}]
    The proof is by counterexample. Consider the reach-avoid specification $\varphi =  \Diamond \Prop_\text{goal} \land \square \Prop_\text{safe}$, which we seek to satisfy under $T = 1$ time steps. Informally, $\varphi$ means ``eventually reach the goal $r_\mathrm{goal}$ region while remaining in $\Xsafe$''. Without loss of generality, let $B^\eta$ be a uniform partition and $x$ be an interior point of $X_0$. Since the partitions become finer as $\eta \to 0$, there exists $\eta_0 > 0$ such that for $\eta \in (0, \eta_0]$ it holds that $\rho^{-\eta}(x) \subseteq X_0$. Define $W_1^\eta := \{w \in W : \mathcal{R}(\rho^{-\eta}(x),a,w)) \subseteq r_\text{goal}\}$ and $C_1^\eta := \{c \in C^\eta : \mathcal{R}(\rho^{-\eta}(x), a, c) \subseteq r_\text{goal}\}$. Note that it is possible to define the partitions $B^\eta$ and $C^\eta$ and the system dynamics be such that $P_W(W_1^{\eta}) = P_W(\cup_{c\in C_1^\eta} c) = 0.5$ for all $\eta \in (0, \eta_0]$. By non-contractiveness of $f$ and since the partition $B^\eta$ is uniform, it holds that each reachable set $\mathcal{R}(\rho^{-\eta}(x), a, c)$, with $c\in C^\eta$, always intersects more than one region $b \in B^\eta$. By the first expression in~\eqref{eq:Gamma_poly}, this implies that $\underline P^\eta(s, a, s') = 0$ for all $s' \in S^\eta$, where $s = \rho^\eta(x)$. 
    By the second expression in~\eqref{eq:Gamma_poly}, the total probability of landing in $r_\text{goal}$ in one step is upper bounded by
    \begin{multline}
    \label{eq:aux_bound}
        \sum_{\substack{s'\in S^\eta :\\ \rho^{-\eta}(s') \cap r_\text{goal} \neq\emptyset}} \overline P^\eta(s, a, s') = \sum_{\substack{s'\in S^\eta :\\ \rho^{-\eta}(s') \cap r_\text{goal} \neq\emptyset}}\sum_{c\in C^\eta}P_W(c)\\
        \boldsymbol{1}(\mathcal{R}(\rho^{-\eta}(s), a, c) \cap\rho^{-\eta}(s') \neq \emptyset).
    \end{multline}
    Swapping the order of the sum operator, we observe that, for each $c\in C^\eta$,
    \begin{multline}
        \label{eq:aux_counterexample}
        \sum_{\substack{s'\in S^\eta :\\ \rho^{-\eta}(s') \cap r_\text{goal} \neq\emptyset}} \boldsymbol{1}(\mathcal{R}(\rho^{-\eta}(s), a, c) \cap \rho^{-\eta}(s') \neq \emptyset)\\
        = |\{s' \in S^\eta : \mathcal{R}(\rho^{-\eta}(s), a, c) \cap \rho^{-\eta}(s') \neq \emptyset \land\\
        \land \rho^{-\eta}(s') \cap r_\text{goal} \neq\emptyset\}|.
    \end{multline}
    Note that $\rho^{-\eta}(s') \cap r_\text{goal}$ is equivalent to $\rho^{-\eta}(s') \subseteq r_\text{goal}$, as $B^\eta$ is an $R$-respecting partition. 
    Adding up the terms of the form \eqref{eq:aux_counterexample} for all $c\in C^\eta_1$, we obtain the following lower bound on Equation~\eqref{eq:aux_bound}:
    \begin{align*}
        &\sum_{c\in C_1^\eta} P_W(c)|\{s'\in S^\eta : \mathcal{R}(\rho^{-\eta}(s),a,c) \cap\\
        &\cap
        \rho^{-\eta}(s') \neq \emptyset \land \rho^{-\eta}(s') \subseteq r_\text{goal}\}|\\
        &=\sum_{c\in C_1^\eta} P_W(c)|\{s'\in S^\eta \cap \mathcal{R}(\rho^{-\eta}(s),a,c)\}|\\
        &\ge 2\sum_{c \in C_1^\eta} P_W(c) = 2P(\cup_{c \in C_1^\eta} c) = 1.
    \end{align*}
    It follows that the abstraction admits probability distributions $\gamma, \gamma' \in \Gamma_{s,a}^\eta$ with $\gamma(r_\text{goal}) = 1$ and $\gamma'(r_\text{goal}) = 0$. By Proposition~\ref{prop:RDP}, $\underline p_T(x) = 0$ and $\overline p_T(x) = 1$. Letting $\varepsilon = 0.9$, we note that $\overline p_T^\eta(x) - \overline p_T^\eta(x) > \varepsilon$, which holds for all $\eta \in (0, \eta_0]$. In consequence, we conclude that the abstraction process is not asymptotically optimal. It is also evident that Algorithm~\ref{alg:algo} never terminates, as the condition in the ``$\mathrm{while}$'' loop is never met. Furthermore, as $\phi(\mathrm{Abs}_\mathbb{I}(\Scal, \eta)) \to 0$ is a sufficient condition for asymptotic optimality, absence of asymptotic optimality together with a contrapositive argument yields that $\phi(\mathrm{Abs}_\mathbb{I}(\Scal, \eta)) \not\to 0$, which concludes the proof.

    We remark that it is common practice to not partition the goal state and, instead, considering it as a single state in the abstraction. However, even in this case, Theorem~\ref{thm:optimality_imdp} still holds, but one should adapt the proof to consider reaching the set of states for which the value function $\underline v^\eta_k$ is nonzero, for some $k \le T$ big enough. Then, one could easily show that the probability of reaching this set in one time step is $0$ in the worst case. By the reasoning in \cite{Gracia:L4DC:2025}, if this is the case for a sufficient number of states, RDP would terminate early, yielding $\underline p^\eta_{k+1}(x_0) = 0$ and $\overline p^\eta_{k+1}(x_0) = 1$ for some $x_0 \in \Xsafe$. Then the conclusions in the reach-avoid counterexample would apply, thus proving the theorem.
\end{proof}
\subsection{Additional Empirical Results}
\label{app:additional_results}

Here we discuss in more detail the case studies in Section~\ref{sec:empirical_validation}. First, regarding the temperature regulation case study, its reach-avoid specification is expressed as the \LTLf formula $\varphi_1 = \Diamond \Prop_\mathrm{goal}\land \square\Prop_\mathrm{safe}$.

Next we discuss the $2$-dimensional case study. The system's states are its $2$D position, while the control is the heading angle, taking $8$ values in the range $[-\pi, \pi]$. The  system dynamics are as follows:
\begin{align*}
    x_{t+1} = x_t + v_l\begin{pmatrix}
        \cos(u_t)\\
        \sin(u_t)
    \end{pmatrix} + c_d \begin{pmatrix}
        \tanh(c_m(x_t) w^{(1)}_t)\\\tanh(c_m(x_t) w^{(2)}_t)
    \end{pmatrix},
\end{align*}
where $c_m(x_t) = 1 - 0.4  e^{(-100((x^{(1)}_t - 0.5)^2 + (x^{(2)}_t - 0.5)^2))}$. The disturbance is Gaussian with covariance $\mathrm{diag}(0.01,0.01)$. Note that its effect is highly dependent on the system's position, being smaller in the vicinity of $x = [0.5,0.5]$. The disturbance coefficient is $c_d = 0.25$ and the linear velocity is $v_l = 0.08$. The specification is represented by the \LTLf formula $\varphi_2 = \square\Prop_\mathrm{safe}\land \square(\Prop_\mathrm{water} \to (\neg \Prop_\mathrm{charge} \mathbin{\mathrm{U}} \Prop_\mathrm{carpet}))\land\Diamond(\Prop_\mathrm{charge})$, where $\varphi'\to\varphi'' \equiv \neg \varphi' \lor \varphi''$. The setup and results are shown in Figure~ \ref{fig:2d}, which compares the results of IMDP and SMDP abstractions. The comparison shows a similar trends to our $1$D case study: the bounds in the satisfaction probabilities are tighter in the SMDP case, and refinement improves them. On the other hand, refinement does not seem to improve the bounds in the IMDP case, and the obtained bounds are vacuous, i.e., practically $[0,1]$ in a big portion of the state space.

Figure~\ref{fig:2du} shows the upper bound in the probability of satisfying the specification for the IMDP and SMDP abstraction and for two different discretizations: a coarse one and a fine one.

\begin{figure}[h]
    \centering

    \begin{minipage}{0.49\linewidth}
        \centering
        \includegraphics[width=\linewidth]{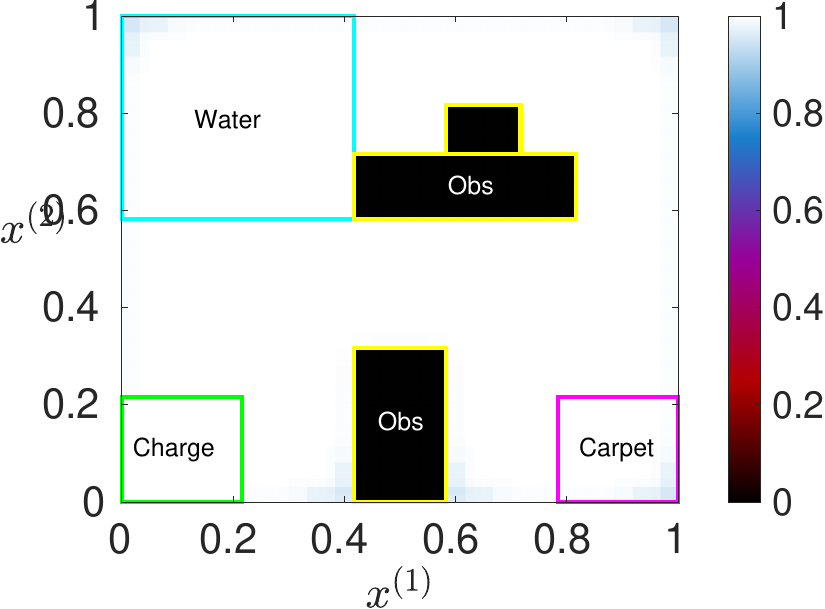}
        \caption{IMDP coarse}
    \end{minipage}
    \hfill
    \begin{minipage}{0.49\linewidth}
        \centering
        \includegraphics[width=\linewidth]{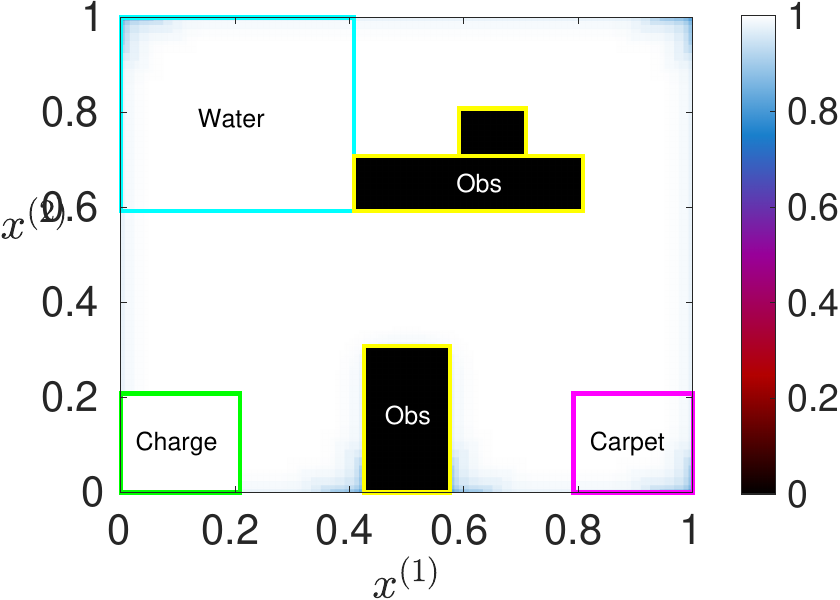}
        \caption{IMDP fine}
    \end{minipage}

    \begin{minipage}{0.49\linewidth}
        \centering
        \includegraphics[width=\linewidth]{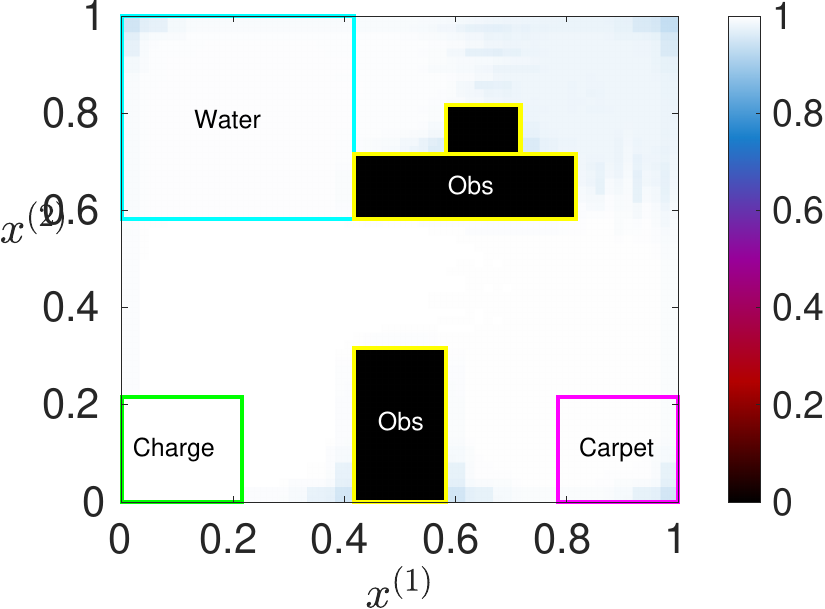}
        \caption{SMDP coarse}
    \end{minipage}
    \hfill
    \begin{minipage}{0.49\linewidth}
        \centering
        \includegraphics[width=\linewidth]{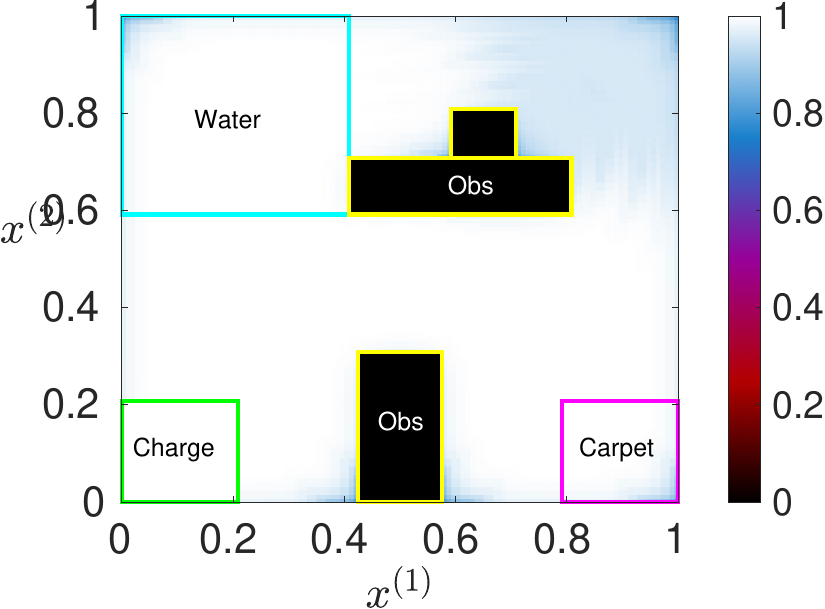}
        \caption{SMDP fine}
    \end{minipage}
    \caption{$2$D case study. Upper bound in the probability of satisfying $\varphi_2$ within $T = 60$ time steps. The color indicates the value of $\overline p_T^\eta(x_0)$.}
    \label{fig:2du}

\end{figure}

\bibliographystyle{IEEEtran}
\bibliography{refs}

\begin{thebibliography}{99}

\bibitem{c1} G. O. Young, ÒSynthetic structure of industrial plastics (Book style with paper title and editor),Ó 	in Plastics, 2nd ed. vol. 3, J. Peters, Ed.  New York: McGraw-Hill, 1964, pp. 15Ð64.
\bibitem{c2} W.-K. Chen, Linear Networks and Systems (Book style).	Belmont, CA: Wadsworth, 1993, pp. 123Ð135.
\bibitem{c3} H. Poor, An Introduction to Signal Detection and Estimation.   New York: Springer-Verlag, 1985, ch. 4.
\bibitem{c4} B. Smith, ÒAn approach to graphs of linear forms (Unpublished work style),Ó unpublished.
\bibitem{c5} E. H. Miller, ÒA note on reflector arrays (Periodical styleÑAccepted for publication),Ó IEEE Trans. Antennas Propagat., to be publised.
\bibitem{c6} J. Wang, ÒFundamentals of erbium-doped fiber amplifiers arrays (Periodical styleÑSubmitted for publication),Ó IEEE J. Quantum Electron., submitted for publication.
\bibitem{c7} C. J. Kaufman, Rocky Mountain Research Lab., Boulder, CO, private communication, May 1995.
\bibitem{c8} Y. Yorozu, M. Hirano, K. Oka, and Y. Tagawa, ÒElectron spectroscopy studies on magneto-optical media and plastic substrate interfaces(Translation Journals style),Ó IEEE Transl. J. Magn.Jpn., vol. 2, Aug. 1987, pp. 740Ð741 [Dig. 9th Annu. Conf. Magnetics Japan, 1982, p. 301].
\bibitem{c9} M. Young, The Techincal Writers Handbook.  Mill Valley, CA: University Science, 1989.
\bibitem{c10} J. U. Duncombe, ÒInfrared navigationÑPart I: An assessment of feasibility (Periodical style),Ó IEEE Trans. Electron Devices, vol. ED-11, pp. 34Ð39, Jan. 1959.
\bibitem{c11} S. Chen, B. Mulgrew, and P. M. Grant, ÒA clustering technique for digital communications channel equalization using radial basis function networks,Ó IEEE Trans. Neural Networks, vol. 4, pp. 570Ð578, July 1993.
\bibitem{c12} R. W. Lucky, ÒAutomatic equalization for digital communication,Ó Bell Syst. Tech. J., vol. 44, no. 4, pp. 547Ð588, Apr. 1965.
\bibitem{c13} S. P. Bingulac, ÒOn the compatibility of adaptive controllers (Published Conference Proceedings style),Ó in Proc. 4th Annu. Allerton Conf. Circuits and Systems Theory, New York, 1994, pp. 8Ð16.
\bibitem{c14} G. R. Faulhaber, ÒDesign of service systems with priority reservation,Ó in Conf. Rec. 1995 IEEE Int. Conf. Communications, pp. 3Ð8.
\bibitem{c15} W. D. Doyle, ÒMagnetization reversal in films with biaxial anisotropy,Ó in 1987 Proc. INTERMAG Conf., pp. 2.2-1Ð2.2-6.
\bibitem{c16} G. W. Juette and L. E. Zeffanella, ÒRadio noise currents n short sections on bundle conductors (Presented Conference Paper style),Ó presented at the IEEE Summer power Meeting, Dallas, TX, June 22Ð27, 1990, Paper 90 SM 690-0 PWRS.
\bibitem{c17} J. G. Kreifeldt, ÒAn analysis of surface-detected EMG as an amplitude-modulated noise,Ó presented at the 1989 Int. Conf. Medicine and Biological Engineering, Chicago, IL.
\bibitem{c18} J. Williams, ÒNarrow-band analyzer (Thesis or Dissertation style),Ó Ph.D. dissertation, Dept. Elect. Eng., Harvard Univ., Cambridge, MA, 1993. 
\bibitem{c19} N. Kawasaki, ÒParametric study of thermal and chemical nonequilibrium nozzle flow,Ó M.S. thesis, Dept. Electron. Eng., Osaka Univ., Osaka, Japan, 1993.
\bibitem{c20} J. P. Wilkinson, ÒNonlinear resonant circuit devices (Patent style),Ó U.S. Patent 3 624 12, July 16, 1990. 






\end{thebibliography}

\end{document}